\documentclass[aps,prd,onecolumn,eqsecnum,amsmath,nofootinbib,preprintnumbers,superscriptaddress]{revtex4-2}

\usepackage[dvipsnames]{xcolor}
\usepackage{color,graphicx,float,subfigure,xcolor}
\usepackage{amsfonts,amssymb,mathrsfs,times}
\usepackage{bm}
\usepackage{multirow}
\usepackage{mathtools}
\usepackage{dsfont}
\usepackage{setspace}
\usepackage{amsmath} 
\usepackage{graphicx}
\usepackage[colorlinks,linkcolor=blue,anchorcolor=blue,citecolor=blue]{hyperref}
\usepackage{ulem}
\usepackage{subcaption}
\usepackage{amsthm}
\newtheorem{theorem}{Theorem}
\usepackage[justification=raggedright,singlelinecheck=false]{caption}

\begin{document}
\title{Exceptional lines of Reissner-Nordstr\"{o}m-de Sitter black hole surrounded by a thin shell of matter}
\preprint{\hfill {\small {ICTS-USTC/PCFT-26-55}}}
\date{\today}

\author{Liang-Bi Wu}
\email{wulb@ucas.ac.cn}
\affiliation{School of Fundamental Physics and Mathematical Sciences, Hangzhou Institute for Advanced Study, UCAS, Hangzhou 310024, China}

\author{Yu-Sen Zhou}
\email{zhou\_ys@mail.ustc.edu.cn}
\affiliation{Interdisciplinary Center for Theoretical Study and Department of Modern Physics,\\
University of Science and Technology of China, Hefei, Anhui 230026, China}

\author{Ming-Fei Ji}
\email{jimingfei@mail.ustc.edu.cn}
\affiliation{Interdisciplinary Center for Theoretical Study and Department of Modern Physics,\\
University of Science and Technology of China, Hefei, Anhui 230026, China}

\author{Xia-Yuan Liu}
\email{liuxiayuan@mail.ustc.edu.cn}
\affiliation{Interdisciplinary Center for Theoretical Study and Department of Modern Physics,\\
University of Science and Technology of China, Hefei, Anhui 230026, China}

\author{Wen-Tao Fu}
\email{fuwentao2024@mail.ustc.edu.cn}
\affiliation{Interdisciplinary Center for Theoretical Study and Department of Modern Physics,\\
University of Science and Technology of China, Hefei, Anhui 230026, China}

\author{Li-Ming Cao}
\email[corresponding author: ]{caolm@ustc.edu.cn}
\affiliation{Interdisciplinary Center for Theoretical Study and Department of Modern Physics,\\
University of Science and Technology of China, Hefei, Anhui 230026, China}
\affiliation{Peng Huanwu Center for Fundamental Theory, Hefei, Anhui 230026, China}

\begin{abstract}
We study the exceptional line (EL) in the quasinormal modes (QNMs) of Reissner-Nordstr\"{o}m-de Sitter black hole surrounded by a static thin shell of matter. For a conformally scalar perturbation, we derive the QNM condition by matching the interior and exterior solutions across the shell and show that higher overtones are particularly sensitive to variations of the shell and background parameters. A mode permutation between two QNMs reveals an exceptional point (EP). After extending the parameter space, this degeneracy forms a continuous EL. We show that the spectral response near the line is intrinsically directional. For a perturbation in parameter space $\epsilon\widehat{\mathbf u}$, the QNM splitting is $|\omega_+-\omega_-| =C_{\widehat{\mathbf u}}\sqrt{\epsilon}+o(\sqrt{\epsilon})$, with $C_{\widehat{\mathbf u}}=(\widehat{\mathbf u}^{T}\mathbf{K}\widehat{\mathbf u})^{1/4}$ and $\mathbf{K}$ is so-called spectral sensitivity anisotropy matrix. The tangent direction of EL is a null direction of $\mathbf{K}$, so that the leading square-root splitting vanishes along the exceptional line, whereas the two principal directions in the normal plane exhibit different sensitivities. Furthermore, the square-root branch structure makes conventional linear QNM parametrizations singular near an EL. We therefore construct an EL adapted parametrization which incorporates both the local geometry of the line and the nonanalytic QNM splitting.
\end{abstract}

\maketitle

\section{Introduction}
When describing the natural oscillations of black holes under perturbation, quasinormal modes (QNMs) manifest as a discrete set of complex-valued frequencies. The real part indicates the oscillation frequency, while the imaginary part reflects the damping rate. As distinctive spectral fingerprints of black holes, QNMs provide a robust means for testing the Kerr hypothesis and probing gravitational dynamics in the strong-field regime. This research framework is known as black hole spectroscopy~\cite{Kokkotas:1999bd,Berti:2009kk,Konoplya:2011qq,Berti:2025hly}.

Recently, more and more advances have emphasized that the non-Hermitian (NH) nature of the QNM eigenvalue problem. Since QNMs are defined by purely ingoing conditions at the horizon and purely outgoing conditions at spatial infinity, the resultant spectrum problem is inherently non-Hermitian~\cite{Ashida:2020dkc}. For NH physics studies in gravitational systems, spectrum instability is an important topic, where initial study on it was provided by~\cite{Nollert:1996rf}. Methods for investigating spectrum instability in the context of QNMs broadly fall into two categories. One of them is modifying the effective potential, where these modifications can be either physically inspired (some environmental effects) or artificially constructed~\cite{Qian:2020cnz,Daghigh:2020jyk,Liu:2021aqh,Li:2024npg,Qian:2024iaq,Xie:2025jbr,Berti:2022xfj,Cheung:2021bol,Yang:2024vor,Courty:2023rxk,Cardoso:2024mrw,Ianniccari:2024ysv,MalatoCorrea:2025iuc,Shen:2025nsl,Wang:2025mxe,Hu:2025efp,Shen:2026hdl,Daghigh:2026lhk,Jia:2026vnx}. What's more, negative deformation of potential will lead to spacetime dynamical instability~\cite{Mai:2025mim,Ma:2026eaf}, where the underlying mechanism is provided in~\cite{Ma:2026eaf}. In addition, small changes in QNM boundary conditions can also lead to significant spectral shifts~\cite{Oshita:2025ibu,Solidoro:2024yxi,Destounis:2025dck,Wu:2025wbp}. The other approach is pseudospectrum analysis which give a qualitative insight of spectrum instability~\cite{Jaramillo:2020tuu,Destounis:2021lum,Cao:2024oud,Cao:2024sot,Arean:2024afl,Garcia-Farina:2024pdd,Arean:2023ejh,Boyanov:2023qqf,Cownden:2023dam,Carballo:2025ajx,Sarkar:2023rhp,Destounis:2023nmb,Luo:2024dxl,Warnick:2024usx,Chen:2024mon,Boyanov:2022ark,Siqueira:2025lww,Besson:2024adi,dePaula:2025fqt,Cai:2025irl,Cao:2025qws,Zhou:2025xdo,Jaramillo:2022kuv}.

NH physics is characterized not only by spectral instability but also by phenomena such as avoided crossings and exceptional points (EPs)~\cite{Heiss:1999alz,Heiss:2012dx}. In black hole perturbation theory, EPs have attracted increasing attention in recent years. Avoided crossings near EPs can induce resonant excitation between QNMs, non-Hermitian explanation for previously puzzling $(2,2,5)$ and $(2,2,6)$ QNM behaviors in Kerr spectrum~\cite{Motohashi:2024fwt,Lo:2025njp}. EPs have also been identified in massive scalar perturbations of Kerr black holes~\cite{Cavalcante:2024swt,Cavalcante:2024kmy,Cavalcante:2025abr}, in Regge-Wheeler potential with Gaussian bumps~\cite{Yang:2025dbn,Cao:2025afs}, in scalar QNMs of hairy black holes~\cite{Cheng:2026gxu}, in Kerr-dS at $(j,\Lambda)\simeq(0.896,0.034)$~\cite{Oshita:2025ibu}, and in Schwarzschild-dS with modified QNM boundary conditions~\cite{Wu:2025wbp}, as well as in dynamical Chern-Simons gravity~\cite{Hu:2025efp} and de Sitter braneworlds~\cite{Jia:2026ncd}. The enhanced spectral sensitivity near EPs has further been characterized through pseudospectra~\cite{Cao:2025afs}. In the time domain, the usual decomposition of the ringdown signal into a sum of independent QNMs can break down near an EP. At a second-order EP, a contribution proportional to $t\mathrm{e}^{-\mathrm{i}\omega t}$ appears~\cite{Yang:2025dbn,PanossoMacedo:2025xnf,Cheng:2026gxu,Imafuku:2026rpn}. More recently, Ref. \cite{Wu:2026chs} showed that nearly degenerate QNMs over a finite observation window are more naturally described by a common carrier mode supplemented by a jet-like correction proportional to $t\mathrm{e}^{-\mathrm{i}\omega t}$. This finite-time description was further investigated in Ref. \cite{Imafuku:2026rpn}, which further developed a Bayesian analysis. In higher-dimensional parameter spaces, EPs can form ELs, which have only recently been explored in gravitational systems~\cite{Cao:2025afs,Nakamoto:2026lyo,Cavalcante:2026vgr}, including Kerr-Newman black holes~\cite{Cavalcante:2026vgr}. Such ELs exhibit an intrinsically directional spectral response. Perturbations tangent to the EL leave the coalesced modes unsplit at leading order, whereas perturbations in the normal plane produce the characteristic $\epsilon^{1/2}$ splitting of a second-order EP, with the splitting strength depending on the perturbation direction~\cite{Cao:2025afs}.

In realistic astrophysical situations, black holes are not completely isolated systems. The surrounding environment, including dark matter distributions, accretion matter, scalar fields, and other forms of matter, can modify the effective geometry around the black hole and consequently leave imprints on the QNM spectrum. Thin shell models provide a simple yet physically motivated framework to describe environmental modifications of black holes. A thin shell can be regarded as an idealized representation of a localized matter distribution surrounding the black hole, separating two different spacetime regions while preserving analytical tractability. The theoretical foundation of thin shells was established by \textit{Israel}, who related the discontinuity of the extrinsic curvature across the shell to its surface stress-energy tensor through the junction conditions~\cite{Israel:1966rt}. The characteristics of black hole images in the presence of thin shells have been presented in~\cite{Cao:2025qzy,Li:2026tkf}. 

Recently, the ringdown of the Schwarzschild black hole surrounded by a thin shell has been studied in~\cite{Laeuger:2025zgb}. The presence of thin shells can also lead to spectrum instability. However, most previous studies have focused on the spectral shifts induced by environmental modifications, while the non-Hermitian structures associated with environmental effects remain largely unexplored. In this work, we investigate the QNM spectra and EL of Reissner-Nordstr\"{o}m-de Sitter (RN-dS) black hole surrounded by a thin shell. For simplicity, we consider a conformal scalar field which does not interact with the shell. This setup provides a natural framework for studying environment-induced NH phenomena: the shell modifies the effective spacetime geometry and introduces additional independent parameters, enlarging the parameter space in which exceptional points can emerge and extend into exceptional lines. Although Ref. \cite{Cao:2025afs} identified the anisotropic spectral response near an EL, a systematic measure of its directional dependence remains lacking. In this work, we use the thin shell of matter model in RN-dS black hole  to quantitatively characterize the anisotropy spectral response associated with EL.

The remainder of this work is organized as follows. In Sec. \ref{sec: sets_up}, we give a brief view on Reissner-Nordstr\"{o}m-de Sitter black hole surrounded by a thin shell of matter. In Sec. \ref{sec: EPs}, we give the condition to determine the QNM spectra and to find the EP. In Sec. \ref{sec: ELs}, the EL finding methods are given and directional spectral sensitivity for the EL is defined. Last, in Sec. \ref{sec: parameterization_EL}, parameterized QNMs formulation near the exception line is introduced. Sec. \ref{conclusions} is the conclusions and discussion. In Appendix \ref{app: Israel junction conditions}, Israel junction conditions are given. In Appendix \ref{app: Heun_equation}, we use the Heun's function formulation to solve the master equation. In Appendix \ref{thm: maximum_sensitivity_direction}, the so-called ``Maximum spectral sensitivity direction'' theorem is proved. In Appendix \ref{app: two_eigenvalues}, we give two eigenvalues of matrix $\mathbf{K}$ (see Eq. (\ref{eq:def_K})) with their numerical evaluations.

\section{Sets up}\label{sec: sets_up}
The $4$-dimensional Reissner-Nordstr\"{o}m-de Sitter (RN-dS) black hole is a solution of Einstein's equations. It describes an electrically charged black hole in a de-Sitter background. In this study, we focus on the case with a static thin shell exists in the RN-dS spacetime. Despite the presence of the thin shell, the spacetime on each side remains locally RN-dS, with the interior and exterior metrics given by~\cite{McManus:2020lgm,Kaplan:2018dqx}
\begin{eqnarray}\label{metric_in}
    \mathrm{d}s^2_{\text{in}}=-f_{\text{in}}(r)\mathrm{d}t^2_{\text{in}}+\frac{\mathrm{d}r^2}{f_{\text{in}}(r)}+r^2(\mathrm{d}\theta^2+\sin^2\theta\mathrm{d}\phi^2)\, ,\quad f_{\text{in}}(r)=-\Lambda\frac{(r-r_{n,\text{in}})(r-r_{-,\text{in}})(r-r_{+,\text{in}})(r-r_\text{c,in})}{3r^2}\, ,
\end{eqnarray}
and
\begin{eqnarray}\label{metric_out}
    \mathrm{d}s^2_{\text{out}}=-f_{\text{out}}(r)\mathrm{d}t^2_{\text{out}}+\frac{\mathrm{d}r^2}{f_{\text{out}}(r)}+r^2(\mathrm{d}\theta^2+\sin^2\theta\mathrm{d}\phi^2)\, ,\quad f_{\text{out}}(r)=-\Lambda\frac{(r-r_{n,\text{out}})(r-r_{-,\text{out}})(r-r_{+,\text{out}})(r-r_\text{c,out})}{3r^2}\, ,
\end{eqnarray}
where $r_{n,\text{in}}=-(r_{-,\text{in}}+r_{+,\text{in}}+r_\text{c,in})$ and $r_{n,\text{out}}=-(r_{-,\text{out}}+r_{+,\text{out}}+r_\text{c,out})$. Original RN-dS black hole allows for three horizons. It is instructive to use the roots of $f(r)$ to characterize the black hole. Therefore, two metric functions in Eq. (\ref{metric_in}) and Eq. (\ref{metric_out}) have written into the form of factorization~\cite{Miguel:2020uln}. The mass and charge parameters $M$ and $Q$ related to these three horizons are 
\begin{eqnarray}\label{M_function_and_Q2_function}
    M=\frac{(r_\text{c}+r_{-}) (r_\text{c}+r_{+}) (r_{-}+r_{+})}{2\Big[r_\text{c}^2+r_\text{c} (r_{-}+r_{+})+r_{-}^2+r_{-} r_{+}+r_{+}^2\Big]}\, ,\quad Q^2=\frac{r_\text{c}r_{-} r_{+} (r_\text{c}+r_{-}+r_{+})}{r_\text{c}^2+r_\text{c} (r_{-}+r_{+})+r_{-}^2+r_{-}r_{+}+r_{+}^2}\, .
\end{eqnarray}
The cosmological constant is
\begin{eqnarray}\label{cosmological_constant}
    \Lambda=\frac{3}{r_\text{c}^2+r_\text{c} (r_{-}+r_{+})+r_{-}^2+r_{-}r_{+}+r_{+}^2}\, .
\end{eqnarray}
We will not explicitly state ``in'' and ``out'' for the sake of simplicity in the above two equations. However, it should be noted that, since the cosmological constant is founded on theoretical grounds, we have no reason to allow it to differ on either side of the shell. In other words, Eq. (\ref{cosmological_constant}) imposes a constraint on the six quantities $r_{-,\text{in}}$, $r_{+,\text{in}}$, $r_\text{c,in}$, $r_{-,\text{out}}$, $r_{+,\text{out}}$ and $r_\text{c,out}$. The time coordinate $t_{\text{in}}$ in the interior of the shell is different from the time coordinate $t_{\text{out}}$ in the exterior. The connection between the temporal coordinates on the interior and exterior is derived by mandating the continuity of the metric across the shell positioned at $r=a$, where $a$ is a constant. This requires
\begin{eqnarray}\label{relation_between_tin_tout}
    \mathrm{d}t_{\text{in}}=\eta(a)\mathrm{d}t_{\text{out}}\, ,\quad \eta(a)\equiv\sqrt{\frac{f_{\text{out}}(a)}{f_{\text{in}}(a)}}\, .
\end{eqnarray}
As for constraints on the two side metric coefficients which are related to the energy momentum tensor on the shell, they can be obtained by using the Israel junction conditions~\cite{Israel:1966rt} (see Appendix \ref{app: Israel junction conditions}). 

So far, we have provided the geometric information inside and outside the thin shell. Now, we will consider the wave equation for a conformal scalar field, namely 
\begin{eqnarray}\label{equation_conformal_scalar_field}
    \Big(\square-\frac{1}{6}R\Big)\Phi=0\, ,
\end{eqnarray}
where $R$ is the Ricci scalar of the RN-dS spacetime, with $R=4\Lambda$. Suppose that $\Phi(t,r,\theta,\phi)=\Psi(t,r)Y_{\ell m}(\theta,\phi)/r$, we have the master perturbation equation in the time domain,
\begin{eqnarray}\label{master_equation_time_domain}
    \Big[\frac{\partial^2}{\partial t^2}-\frac{\partial^2}{\partial x^2}+V(x)\Big]\Psi(t,x)=0\, ,\quad \mathrm{d}x=\frac{\mathrm{d}r}{f(r)}\, ,
\end{eqnarray}
where the potential is~\cite{Ahmed:2016lou}
\begin{eqnarray}\label{effective_potential}
    V(r)=f(r)\Big[\frac{\ell(\ell+1)}{r^2}+\frac{2M}{r^3}-\frac{2Q^2}{r^4}\Big]\, .
\end{eqnarray}
In the frequency domain, adopting the Fourier ansatz $\mathrm{e}^{-\mathrm{i}\omega t}$, the perturbation equation (\ref{master_equation_time_domain}) can be reduced to the following Schr\"{o}dinger-like equation
\begin{eqnarray}\label{master_equation_frequency_domain}
    \Big[\frac{\mathrm{d}^2}{\mathrm{d}x^2}+\omega^2-V(x
    )\Big]\Psi(x)=0\, .
\end{eqnarray}
Inside and outside the thin shell, the coordinate times $t_{\text{in}}$ and $t_{\text{out}}$ are not synchronized. From Eq. (\ref{relation_between_tin_tout}), if a mode oscillates with frequency $\omega_{\text{in}}$ with respect to $t_{\text{in}}$, an outside observer will measure a different frequency $\omega_{\text{out}}$:
\begin{eqnarray}
    \mathrm{e}^{-\mathrm{i}\omega_{\text{in}}\cdot t_{\text{in}}}=\mathrm{e}^{-\mathrm{i}\omega_{\text{in}}\eta(a) \cdot t_{\text{out}}}=\mathrm{e}^{-\mathrm{i}\omega_{\text{out}}\cdot t_{\text{out}}}\quad \Rightarrow \quad \omega_{\text{out}}=\eta(a) \omega_{\text{in}}\, ,
\end{eqnarray}
This frequency difference is a key feature of the thin shell model. The shell not only modifies the spacetime geometry but also acts as an interface that rescales the time coordinate. As a result, the oscillation frequencies inside and outside are not the same. This phenomenon is explained as gravitational redshift~\cite{Laeuger:2025zgb}. Both inside and outside analytical solutions for the potential (\ref{effective_potential}) of the RN-dS black hole can be derived by the Heun's function method~\cite{Hatsuda:2020sbn,Noda:2022zgk,Chen:2025sbz,Chen:2024rov,Xia:2025hwt,Li:2026zsg,Mi:2025fbt,Li:2025lgn,Jiang:2025mcj,Oshita:2021iyn} (see more details in Appendix \ref{app: Heun_equation}). More importantly, in order to ensure that the influence of the shell is not trivial, considering that QNM problem usually study the region between the event horizon and the cosmological horizon, the shell position should satisfy the following condition 
\begin{eqnarray}\label{shell_position}
    \text{max}\big\{r_{+,\text{in}},r_{+,\text{out}}\big\}<a<\text{min}\big\{r_\text{c,in},r_\text{c,out}\big\}\, .
\end{eqnarray}
Therefore, the event horizon of the new black hole constructed through a thin spherical shell is $r_{+,\text{in}}$, the cosmological horizon is $r_\text{c,out}$, and the Cauchy horizon is still $r_{-,\text{in}}$. For $r_{+,\text{in}}<r<a$, the general solution is denoted by $C_{\text{in}}\cdot\Psi_{\text{in}}(\omega_{\text{in}},r)$ (see Eq. (\ref{Psi_in})) and for $a<r<r_{\text{c,out}}$, the general solution is denoted by $C_{\text{out}}\cdot\Psi_{\text{out}}(\omega_{\text{out}},r)$ (see Eq. (\ref{Psi_out})), where $C_{\text{in}}$ and $C_{\text{out}}$ are general constants. In this study, we will use the unit such that $r_{+,\text{in}}=1$.

\section{Quasinormal modes and exceptional points}\label{sec: EPs}
In this section, we study QNMs and the EP of the Reissner-Nordstr\"{o}m-de Sitter black hole surrounded by a thin shell of matter. Requiring the model to be simple while maintaining its physical properties, we consider a simple case where the scalar field and the thin shell have no coupling for the sake of solving the QNM spectra. Therefore, the conditions on the shell of $\Phi$ are $[[\Phi]]=0$ and $[[n^\mu\partial_\mu\Phi]]=0$ ($n^\mu$ is unit normal vector of the shell [see Eq. (\ref{normal_vectors})].), from which one gets the conditions in terms of the master variable $\Psi$ as follow
\begin{eqnarray}\label{QNMs_condition_eta_general}
    \det\mathcal{M}(\omega_{\text{in}})\equiv\det
    \begin{bmatrix}
        \Psi_{\text{out}}(\omega_{\text{out}},a) & -\Psi_{\text{in}}(\omega_{\text{in}},a)\\
        \eta(a)(\partial_r\Psi_{\text{out}})(\omega_{\text{out}},a)-\frac{\eta(a)}{a}\Psi_{\text{out}}(\omega_{\text{out}},a)&-(\partial_r\Psi_{\text{in}})(\omega_{\text{in}},a)+\frac{1}{a}\Psi_{\text{in}}(\omega_{\text{in}},a)
    \end{bmatrix}=0\, ,
\end{eqnarray}
where $\omega_{\text{out}}=\eta(a)\omega_{\text{in}}$. Equation (\ref{QNMs_condition_eta_general}) is derived from the condition that there are non-zero solutions to the linear equations of $C_{\text{out}}$ and $C_{\text{in}}$. For the case where the thin shell does not exist, we obtain $\eta=1$. So in such case, Eq. (\ref{QNMs_condition_eta_general}) will reduce to the QNM condition of a single RN-dS spacetime.  It means that the QNM spectra will be solved by the condition of linear dependence between $\Psi_{\text{in}}(\omega,r)$ and $\Psi_{\text{out}}(\omega,r)$, i.e., 
\begin{eqnarray}\label{QNMs_condition_eta_1}
    \mathcal{W}[\Psi_{\text{in}}(\omega,r),\Psi_{\text{out}}(\omega,r)]\Big|_{r=a}=0\, ,\quad \omega_{\text{in}}=\omega_{\text{out}}=\omega\, ,
\end{eqnarray}
in which the above Wronskian is defined as $\mathcal{W}[u(r),v(r)]\equiv u\mathrm{d}v/\mathrm{d}r-v\mathrm{d}u/\mathrm{d}r$. The zeros of function $\det\mathcal{M}(\omega_{\text{in}})$, i.e., the QNM spectra, are obtained by scanning the complex $\omega$ plane and then using the built-in function \textit{FindRoot} in \textit{Mathematica}~\cite{Wu:2025sbq,Wu:2025wbp,Wu:2026hvf}. For the sake of simplicity in notation, $\omega$ refers to $\omega_{\text{in}}$ in the followings.

We have $5$ freely adjustable parameters including $r_{\text{c},\text{in}}$, $r_{-,\text{in}}$, $r_{+,\text{out}}$, $r_{\text{c},\text{out}}$, and $a$. Here, we focus on the case with angular momentum number $\ell=2$. In order to find the EP, two horizons $r_{-,\text{in}}$ and $r_{\text{c,in}}$ of the inside spacetime are fixed as $r_{-,\text{in}}=0.8$ and $r_{\text{c,in}}=3$ in the following. In Fig. \ref{QNM_Spectra_rcout_migration}, we show the QNM spectra migrations with respect to the migration of $r_{\text{c,out}}$, where for the left panel, $a=1.25$ while for the right panel, $a=2$. As $r_{\text{c,out}}=r_{\text{c,in}}$, from Eq. (\ref{relation_between_tin_tout}), one gets $\eta=1$, which means that there is no thin shell. In other words, the QNM spectra degenerate into the QNM spectra of conformal scalar field in RN-dS black hole. So such spectra will not be influenced by the position of the thin shell, namely $a$. The QNM spectra results with no thin shell are marked as star in Fig. \ref{QNM_Spectra_rcout_migration}. The overtone numbers are assigned according to their imaginary parts at this reference parameter point and retained along the corresponding branches as the parameters vary. Furthermore, two boundaries of $r_{\text{c,out}}$ namely $r_{\text{c,out}}=2.853$ and $r_{\text{c,out}}=3.511$ comes from the requirements that the cosmological constant remains consistent in both inside and outside spacetime. For each panels, it can be found that for higher-order overtones, their migration distances are longer, indicating the enhanced spectrum instability of the overtone modes~\cite{Jaramillo:2020tuu}. Comparing the two panel, we can see the impact of different $a$ values on spectrum migrations. For our parameter setting, the QNM spectra for $a=2$ are more stable than those for $a=1.25$.

\begin{figure}[htbp]
    \centering
   \subfigure[]{\includegraphics[width=0.47\linewidth]{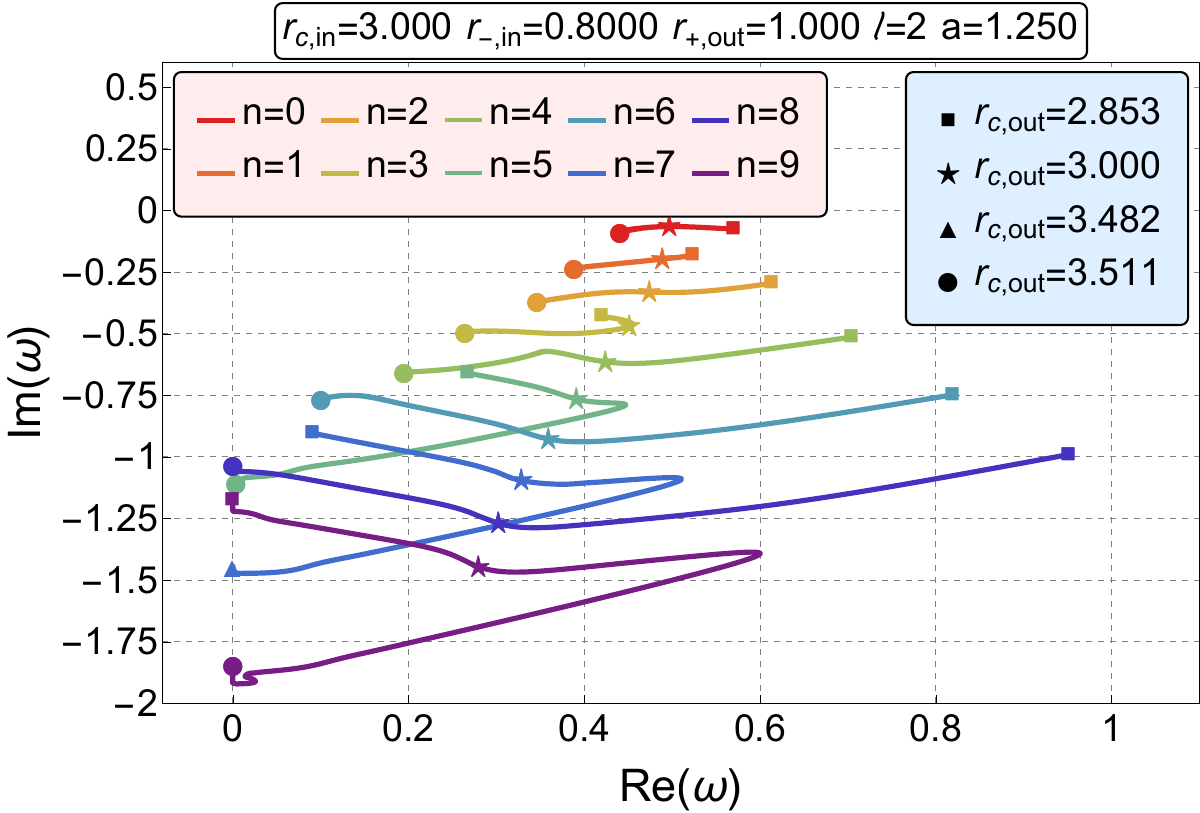}\label{QNM_Spectra_rcout_migration_1}}\hfill
   \subfigure[]{\includegraphics[width=0.47\linewidth]{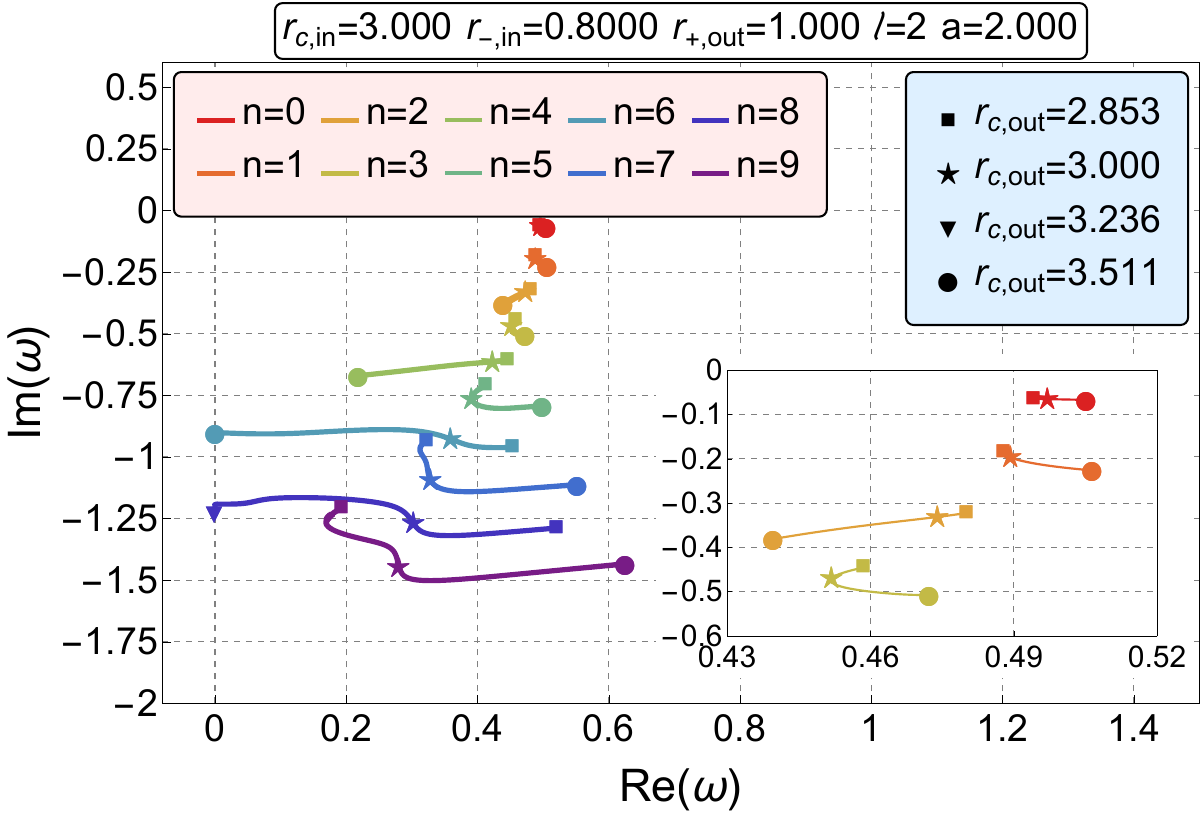}\label{QNM_Spectra_rcout_migration_2}}
    \caption{QNM spectra migration induced by varying the outer cosmological horizon $r_{\text{c,out}}$. The left and right panels correspond to two different shell locations, $a=1.25$ and $a=2.0$, respectively. The fixed parameters are $r_{-,\text{in}}=0.8$, $r_{\text{c,in}}=3$, $r_{+,\text{out}}=1$, and $\ell=2$. Different symbols denote different values of $r_{\text{c,out}}$, while different curves represent the trajectories of QNM overtones during the parameter variation. The QNM spectrum without a thin shell is marked by stars.}
    \label{QNM_Spectra_rcout_migration}
\end{figure}

We can not only study the impact of $r_{\text{c,out}}$ on spectra, but also pay attention to the impact of $a$ on spectra. In Fig. \ref{QNM_a_migration}, the migration of QNM spectra due to the change of $a$ is depicted, where $r_{\text{c,out}}$ is given by $2.853$. The higher overtones are also more unstable than lower overtones. Note that the start QNM spectra with $a=1.25$ come from the results in Fig. \ref{QNM_Spectra_rcout_migration_1}, which are presented by the solid squares. The color data of such start spectra in Fig. \ref{QNM_a_migration} is \textit{BrightBands} in \textit{Mathematica}. Then one uses Eq. (\ref{QNMs_condition_eta_general}) to obtain QNM spectra migrations for each overtones. The end QNM spectra are label by the hollow squares with the color being \textit{BrightBands} too. Consequently, an interesting yet reasonable phenomenon occurred. When the QNM results with $a=2$ and $r_{\text{c,out}}=2.853$ of Fig. \ref{QNM_Spectra_rcout_migration_2} are also placed in Fig. \ref{QNM_a_migration} where the color is given by \textit{Rainbow}, the mode for $n=6$ migrates to that of $n=7$, and the mode for $n=7$ migrates to that of $n=6$. This phenomenon indicates the existence of an exceptional point (EP) in the parameter space $(a, r_{\text{c,out}})$.

\begin{figure}[htbp]
    \centering
   \includegraphics[width=0.80\linewidth]{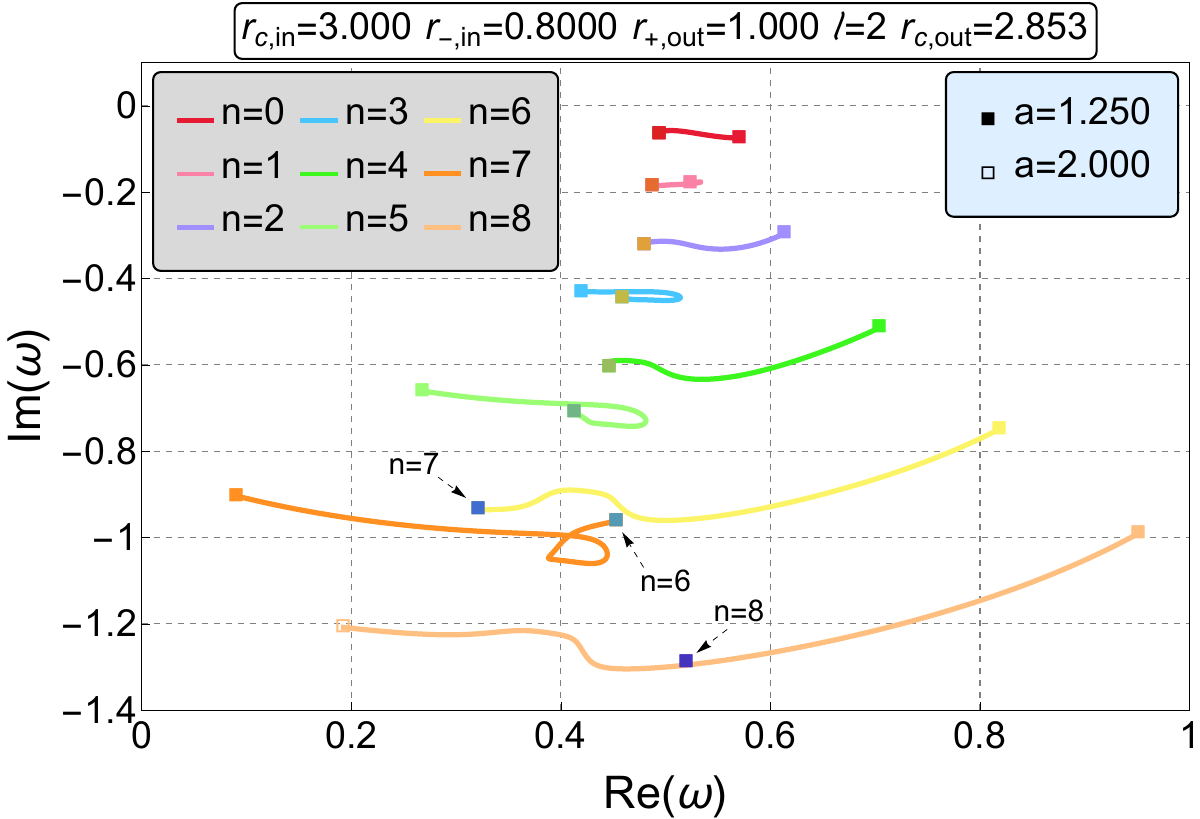}
    \caption{Migration of the QNM spectra caused by varying the shell radius $a$ for fixed $r_{\text{c,out}}=2.853$. The parameters are chosen as $r_{-,\text{in}}=0.8$, $r_{\text{c,in}}=3$, $r_{+,\mathrm{out}}=1$, and $\ell=2$. The solid squares and hollow squares denote spectra corresponding to $a=1.25$ and $a=2.0$, respectively. The trajectories illustrate the continuous evolution of the QNM spectra during the shell radius variation. In particular, the exchange of the $n=6$ and $n=7$ modes between the two parameter choices signals the presence of an EP in the parameter space $(a,r_{\text{c,out}})$.}
    \label{QNM_a_migration}
\end{figure}

From the above statement, it has been demonstrate that there exist an EP in the parameter region constructed by $(a, r_{\text{c,out}})$, and such region is represented by the gray area in Fig. \ref{QNM_Spectra_Exchange_2}. By calculating the spectra along the parameter trajectory $(A\to B\to C\to D\to A)$ in Fig. \ref{QNM_Spectra_Exchange_2}, we observe the exchange of mode $n=7$ and mode $n=6$. In detail, on the frequency complex plane, the migration trajectory of modes are described in Fig. \ref{QNM_Spectra_Exchange_1}, in which one is in red color and the other is in green color. Although the EP exists in the gray region, its exact location cannot be determined at present. This requires us to provide a reliable algorithm to determine the location of EP.

\begin{figure}[htbp]
    \centering
   \subfigure[]{\includegraphics[width=0.47\linewidth]{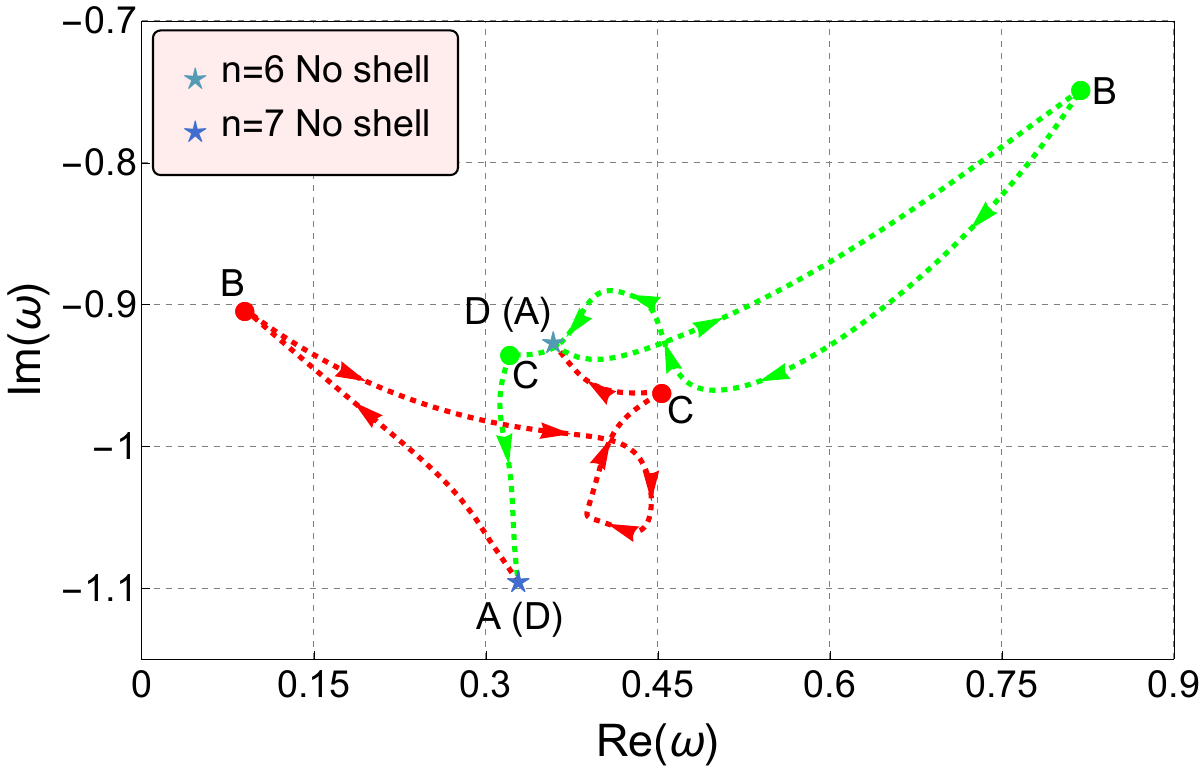}\label{QNM_Spectra_Exchange_1}}\hfill
   \subfigure[]{\includegraphics[width=0.47\linewidth]{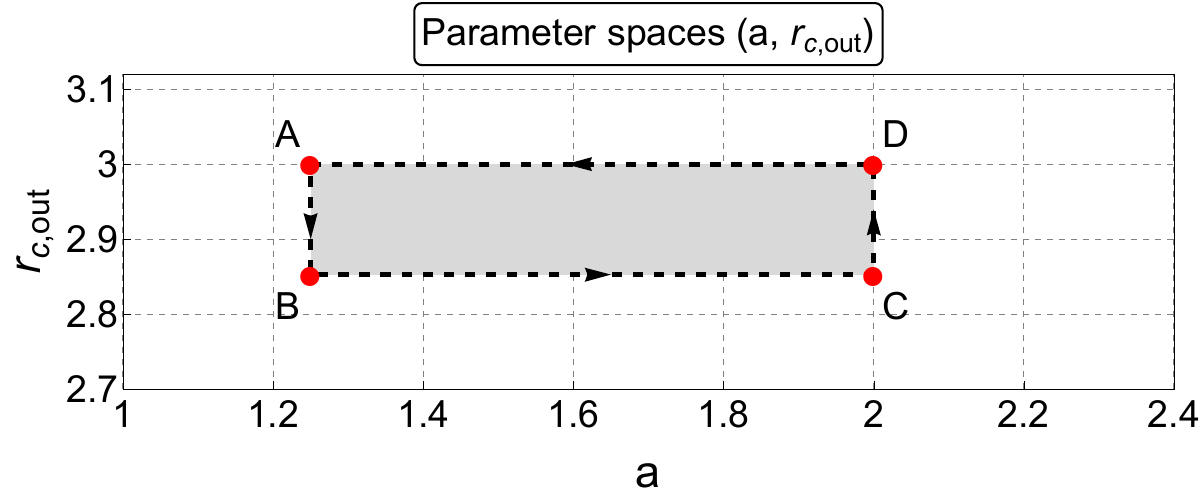}\label{QNM_Spectra_Exchange_2}}
    \caption{Identification of an EP in the parameter space $(a,r_{\text{c,out}})$. (a) Trajectories of the two QNM branches ($n=6$ and $n=7$) in the complex-frequency plane along a closed loop $A\rightarrow B\rightarrow C\rightarrow D\rightarrow A$. The two modes exchange their positions after one encircling cycle, demonstrating the mode permutation associated with an exceptional point. (b) The corresponding closed trajectory in the parameter space $(a,r_{\text{c,out}})$.}
    \label{QNM_Spectra_Exchange}
\end{figure}

Now, we give a practical approach to find the EP. First, the method starts from a rectangular region in a two-dimensional parameter space. The QNM spectra are calculated along the boundary of this rectangle. Then the spectra trajectories are tracked. If the spectra return to their original positions, the rectangle is considered to contain no EP. If the spectra are exchanged, the rectangle is considered to contain one EP at least. The above fact has been confirmed for our present two-dimensional parameter space $(a, r_{\text{c,out}})$, see Fig. \ref{QNM_Spectra_Exchange_1}. Once an EP is detected inside a rectangle, the rectangle is subdivided into two smaller rectangles. The subdivision is performed alternately: vertical subdivision into left and right regions, and horizontal subdivision into upper and lower regions. One repeats subdivision until the desired number of iterations. Fig. \ref{EP_Finding} shows such subdivision. To be more detailed, parameters start from the upper left corner of the rectangle denoted by black star, rotate counterclockwise, and one tracks the trajectories of two QNM spectra. 

For the case in Fig. \ref{EP_Finding_LR}, we additionally evaluate the QNM spectra along a trajectory extending from the midpoint of the lower edge to the midpoint of the upper edge, as indicated by the red arrow in Fig. \ref{EP_Finding_LR}. If the outcome at the upper edge midpoint coincides with that obtained on the original rectangular contour at the same point, the EP resides on the left side of the rectangle. The left half of the rectangle then constitutes the new rectangle, and the QNM results on the right side of this new rectangle are supplanted by the result of the additional trajectory. Conversely, should the outcome at the upper edge midpoint prove discordant with that of the original contour, the EP is situated on the right side. The right half becomes the new rectangle, and the QNM results on its left side are superseded by the inverse results of the additional trajectory (red dashed arrow in Fig. \ref{EP_Finding_LR}). For the scenario depicted in Fig. \ref{EP_Finding_TB}, we also evaluate the QNM spectra along a trajectory extending from the midpoint of the right edge to the midpoint of the left edge, as indicated by the red arrow in Fig. \ref{EP_Finding_TB}. If the result at the midpoint of the left edge coincides with that obtained on the original rectangular contour at the same point, the EP is situated on the upper side of the rectangle. The upper half of the rectangle then constitutes the new rectangle, and the QNM result on the lower side of this new rectangle is supplanted by the inverse result of the additional trajectory (red dashed arrow in Fig. \ref{EP_Finding_TB}). Conversely, should the result at the midpoint of the left edge prove discordant with that of the original contour, the exceptional point resides on the lower side. The lower half becomes the new rectangle, and the QNM result on its upper side is superseded by the result of the additional trajectory. It is worth mentioning that the number of horizontal and vertical divisions does not necessarily have to be the same. For more detail, one can refer to~\cite{deguchi2024}, in which the study object is an elastic system.

\begin{figure}[htbp]
    \centering
   \subfigure[]{\includegraphics[width=0.47\linewidth]{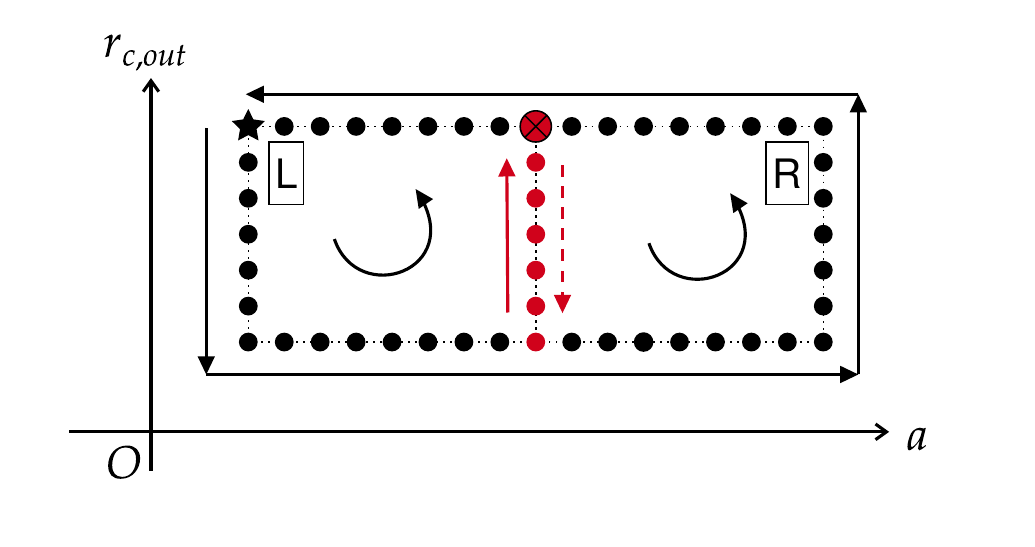}\label{EP_Finding_LR}}\hfill
   \subfigure[]{\includegraphics[width=0.47\linewidth]{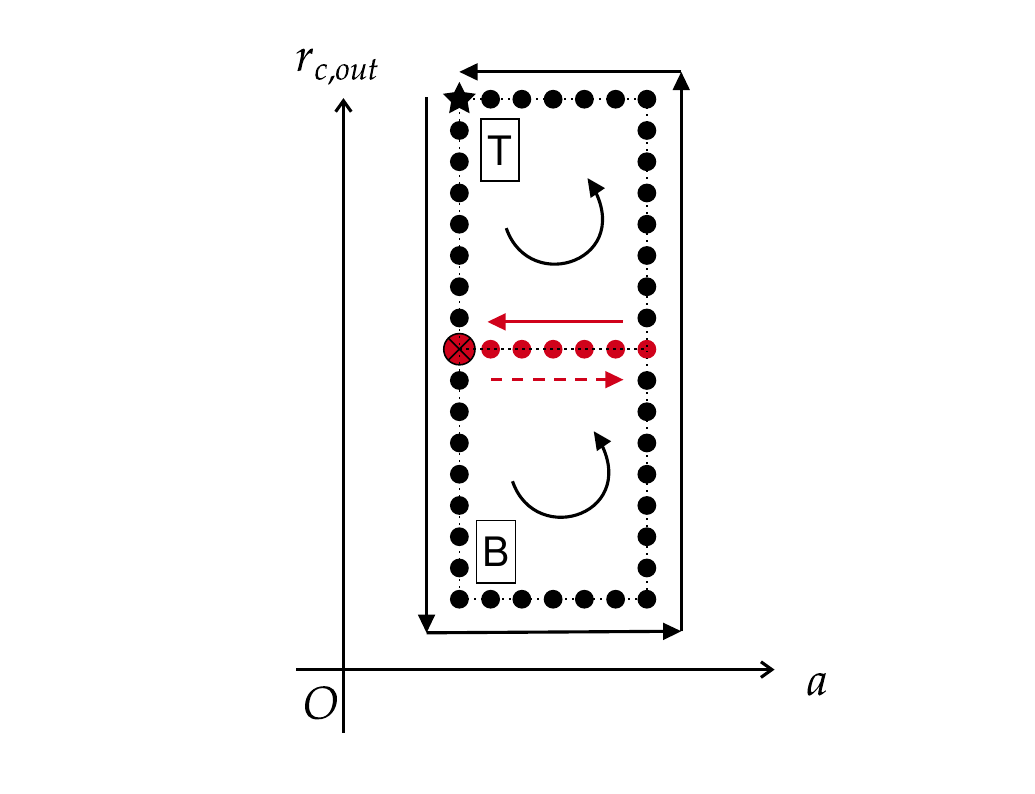}\label{EP_Finding_TB}}
    \caption{The schematic diagram for finding exceptional points. The left panel stands for the left and right subdivision, while the right panel stands for the top and bottom subdivision. Continuously perform these two operations to find EPs, namely compressing space for EP. The number of these two operations may not be the same.}
    \label{EP_Finding}
\end{figure}

When the small rectangular contours obtained by subdivision are sufficiently small, we can directly use the spectra results on the rectangle to obtain the positions of the EP. The method called octagon method (nine points) is introduced as follows~\cite{Feldmaier:2016ukr}. We denote that two selected eigenvalues $\omega_{+}$ and $\omega_{-}$, and the two parameter space is set by $(\mu,\nu)$. At an EP, it is known that the eigenvalue is not analytical but has a Puiseux series structure for some parameter~\cite{kato1976perturbation}. In order to circumvent the branch point singularity at EP, two auxiliary functions are introduced~\cite{Feldmaier:2016ukr,Nennig_2020,Uzdin_2010,PhysRevA.79.053408}:
\begin{eqnarray}\label{sum_omega}
    g(\mu,\nu)\equiv\omega_{+}+\omega_{-}=A+B(\mu-\mu_0)+C(\nu-\nu_0)+o(\mu-\mu_0)+o(\nu-\nu_0)\, ,
\end{eqnarray}
and
\begin{eqnarray}\label{squared_difference_omega}
    h(\mu,\nu)\equiv(\omega_{+}-\omega_{-})^2&=&D+E(\mu-\mu_0)+F(\nu-\nu_0)+G(\mu-\mu_0)^2+H(\mu-\mu_0)(\nu-\nu_0)+I(\nu-\nu_0)^2\nonumber\\
    &&+o\big((\mu-\mu_0)^2\big)+o\big((\nu-\nu_0)^2\big)\, .
\end{eqnarray}
By construction, these two functions are analytical. So their power series should exist and are shown explicitly in Eq. (\ref{sum_omega}) and Eq. (\ref{squared_difference_omega}). In Eq. (\ref{sum_omega}) and Eq. (\ref{squared_difference_omega}), complex coefficients $A$ to $I$ are expansion coefficients and the expansion point is $(\mu_0,\nu_0)$. It is worth mentioning that in~\cite{Uzdin_2010}, only linear terms are remained in Eq. (\ref{squared_difference_omega}) as they use the three-point method. Here, we expand the function $h(\mu,\nu)$ into second-order, and this will result in rougher initial parameters being able to find precise exceptional points.

For sake of getting simple relations for the coefficients, one use such nine points in the parameter space $(\mu,\nu)$. Specifically, among these points, four at the corner of a rectangle, four are midpoints of the rectangular edges, and one at the center, where schematic diagram is shown in Fig. \ref{Nine_Points_Method}. For $i=0,\cdots,8$, the sums $g_i=\omega_{+,i}+\omega_{-,i}$ and the squared differences $h_i=(\omega_{+,i}-\omega_{-,i})^2$ can be used to determine nine coefficients in Eq. (\ref{sum_omega}) and Eq. (\ref{squared_difference_omega}). Neglecting higher-order terms, it is not difficult to calculate that
\begin{eqnarray}\label{nine_coefficients}
    &&A=g_0\, ,\quad B=\frac{g_1-g_5}{2(\Delta\mu)}\, ,\quad C=\frac{g_3-g_7}{2(\Delta\nu)}\, ,\nonumber\\
    &&D=h_0\, ,\quad E=\frac{h_1-h_5}{2(\Delta\mu)}\, ,\quad F=\frac{h_3-h_7}{2(\Delta\nu)}\, ,\nonumber\\
    &&G=\frac{h_1+h_5-2h_0}{2(\Delta\mu)^2}\, ,\quad H=\frac{h_2-h_4+h_6-h_8}{4(\Delta\mu)(\Delta\nu)}\, ,\quad I=\frac{h_3+h_7-2h_0}{2(\Delta\nu)^2}\, .
\end{eqnarray}

\begin{figure}[htbp]
    \centering
   \includegraphics[width=0.75\linewidth]{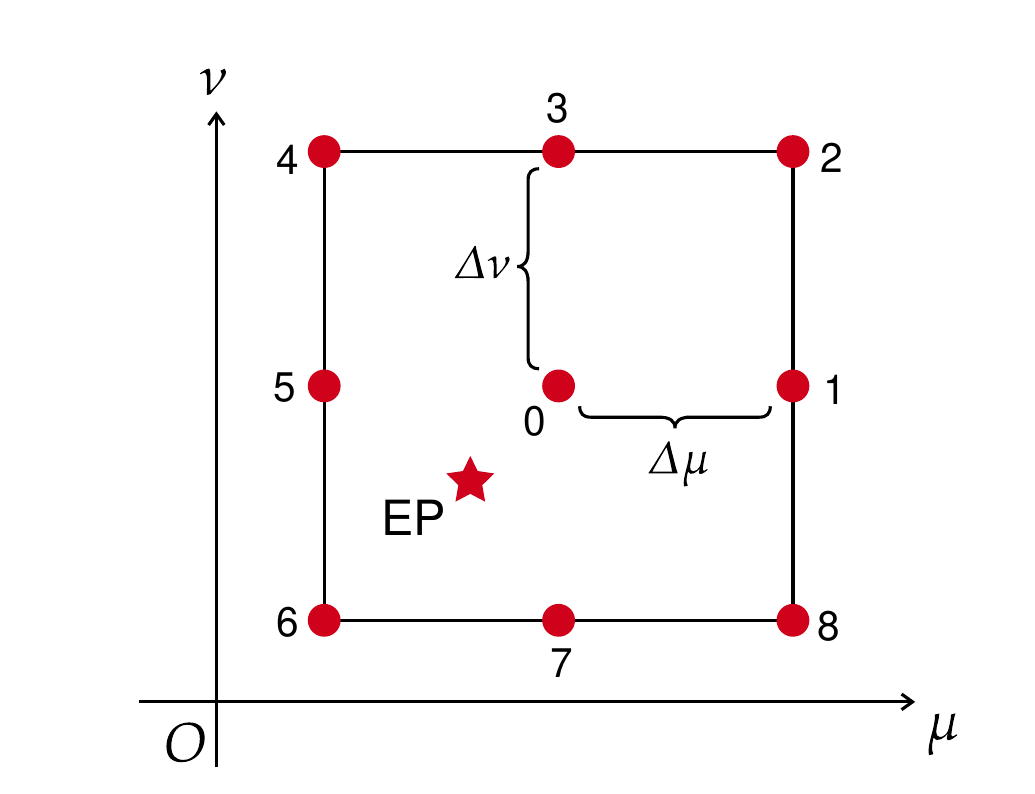}
    \caption{The schematic diagram of the nine point method. One should calculate $\omega_{+}$ and $\omega_{-}$ on these nine points. EP is represented by the red star, and such EP does not necessarily have to be inside the rectangle.}
    \label{Nine_Points_Method}
\end{figure}

Having obtained all coefficients, one can make an estimation for the position $(\mu_{\text{EP}},\nu_{\text{EP}})$ by solving $h(\mu_{\text{EP}},\nu_{\text{EP}})=0$ which is the condition for EP. Define offsets in the $\mu$ and $\nu$ directions as $\delta\mu\equiv\mu_{\text{EP}}-\mu_0$ and $\delta\nu\equiv\nu_{\text{EP}}-\nu_0$, one derives the following algebraic equation
\begin{eqnarray}\label{algebraic_equation}
    0=D+E\delta\mu+F(\delta\nu)+G(\delta\mu)^2+H(\delta\mu)(\delta\nu)+I(\delta\nu)^2\, .
\end{eqnarray}
Solving the real part and imaginary part of above equation, we will get $\delta\mu\in\mathbb{R}$ and $\delta\nu\in\mathbb{R}$ (See Appendix B in~\cite{Feldmaier:2016ukr}). Accordingly, EP will be obtained with $\mu_{\text{EP}}=\mu_0+\delta\mu$ and $\nu_{\text{EP}}=\nu_0+\delta\nu$. Here, for our present model, $\mu=a$ and $\nu=r_{\text{c,out}}$. In order to evaluate nine quantities in Eqs. (\ref{nine_coefficients}), the numerical settings used for $\Delta a$ and $\Delta r_{\text{c,out}}$ are given by $0.000732421875$ and $0.00057421875$, respectively. Finally, one gets $(\mu_{\text{EP}},\nu_{\text{EP}})\simeq(1.969983,2.910025)$ and $\omega_{\text{EP}}\simeq0.3715095-0.9712053\mathrm{i}$.

\section{Exceptional lines and directional spectral sensitivity}\label{sec: ELs}
In this section, we will discuss the exceptional line (EL) and directional spectral sensitivity. First, we introduce definition of the EL. Let the system depend smoothly on a real three-dimensional parameter vector $\mathbf{p}=(p_1,p_2,p_3)^\mathrm{T}=(\mu,\nu,\xi)^\mathrm{T}\in\mathbb{R}^3$. Suppose that two eigenvalues, denoted by $\omega_{+}(\mathbf{p})$ and $\omega_{-}(\mathbf{p})$, coalesce at an exceptional point $\mathbf{p}_0$. We introduce the squared eigenvalue difference which is also called the discriminant function
\begin{eqnarray}\label{eq:def_h}
    h(\mathbf{p}):=\Big[\omega_{+}(\mathbf{p})-\omega_{-}(\mathbf{p})\Big]^2\, .
\end{eqnarray}
The function $h$ is generally complex-valued and may be decomposed as
\begin{eqnarray}\label{eq:h_F1_F2}
    h(\mathbf{p})=F_1(\mathbf{p})+\mathrm{i} F_2(\mathbf{p})\, ,\qquad \text{with}\qquad F_1(\mathbf{p})=\operatorname{Re}h(\mathbf{p})\, ,\qquad F_2(\mathbf{p})=\operatorname{Im}h(\mathbf{p})\, .
\end{eqnarray}
Therefore, in a three-dimensional parameter space, the exceptional line is locally given by the intersection
\begin{eqnarray}\label{eq:EP_line}
    \mathcal{L}_{\text{EP}}=\Big\{
        \mathbf{p}\in\mathbb{R}^{3}\,\Big|\,F_1(\mathbf{p})=0,\;F_2(\mathbf{p})=0
    \Big\}\, .
\end{eqnarray}
From the above equation, we can understand the EL as the intersection of two surfaces. Note that unless otherwise specified, $\mu$ refers to $a$, $\nu$ refers to $r_{\text{c,out}}$, and $\xi$ refers to $r_{+,\text{out}}$. This EL is also written by $\mathbf p_{\rm EL}=\mathbf p_{\rm EL}(s)$, where $s$ is the arc length along the EL. 

In Sec. \ref{sec: EPs}, given $r_{+,\text{out}}=1$, we have found an EP at the parameter space $(a,r_{\text{c,out}})$. In order to achieve the EL passing through such EP, we take the results for $\omega_{+}$ and $\omega_{-}$ at $r_{+,\text{out}}=1$ as the benchmark, and then slowly vary the value of $r_{+,\text{out}}$. For each new $r_{+,\text{out}}$, the corresponding $\omega_{+}$ and $\omega_{-}$ are obtained from the previous $\omega_{+}$ and $\omega_{-}$ using the \textit{FindRoot} function. Thereafter, using Eq. (\ref{sum_omega}), Eq. (\ref{squared_difference_omega}) (see also Eq. (\ref{h_i_mu_nu})) and Eqs. (\ref{nine_coefficients}), six function sequences in terms of $\xi_i$, i.e., $D(\xi_i)$, $E(\xi_i)$, $F(\xi_i)$, $G(\xi_i)$, $H(\xi_i)$, $I(\xi_i)$ can be derived, in which the expanding center is the result of the previous EP and $\Delta\xi$ is fixed. Following the above method, EL will be derived. 

Fig. \ref{EL_Plot_and_omega_EL_rp_out} shows the numerically traced exceptional line with $\Delta\xi=0.0005$ and the corresponding exceptional frequency. As shown in Fig. \ref{EL_Plot}, starting from the point labeled ``Start'', the EP can be continuously traced as the system parameters are varied, eventually reaching the point labeled ``End''. Fig. \ref{omega_EL_rp_out_Plot} shows the evolution of the corresponding degenerate frequency $\omega_{\text{EP}}$ along the exceptional line, where $r_{+,\text{out}}$ is used as the parameter characterizing the position along the line. Both the real and imaginary parts of $\omega_{\text{EP}}$ vary smoothly along the EL. In particular, as $r_{+,\text{out}}$ increases from $1$ to $1.76$, the real part of the exceptional frequency increases monotonically from approximately $\operatorname{Re}(\omega_{\text{EP}})\simeq 0.37$ to $\operatorname{Re}(\omega_{\text{EP}})\simeq 0.55$. At the same time, the imaginary part increases from approximately $\operatorname{Im}(\omega_{\text{EP}})\simeq -0.97$ to $\operatorname{Im}(\omega_{\text{EP}})\simeq -0.89$, and therefore becomes progressively less negative. These numerical results therefore demonstrate that the EP is not an isolated object in the three-dimensional parameter space. Instead, it extends into a continuous EL, along which both the location of the degeneracy in parameter space and the associated degenerate frequency $\omega_{\text{EP}}$ evolve continuously.

\begin{figure}[htbp]
    \centering
    \subfigure[]
   {\includegraphics[width=0.47\linewidth]{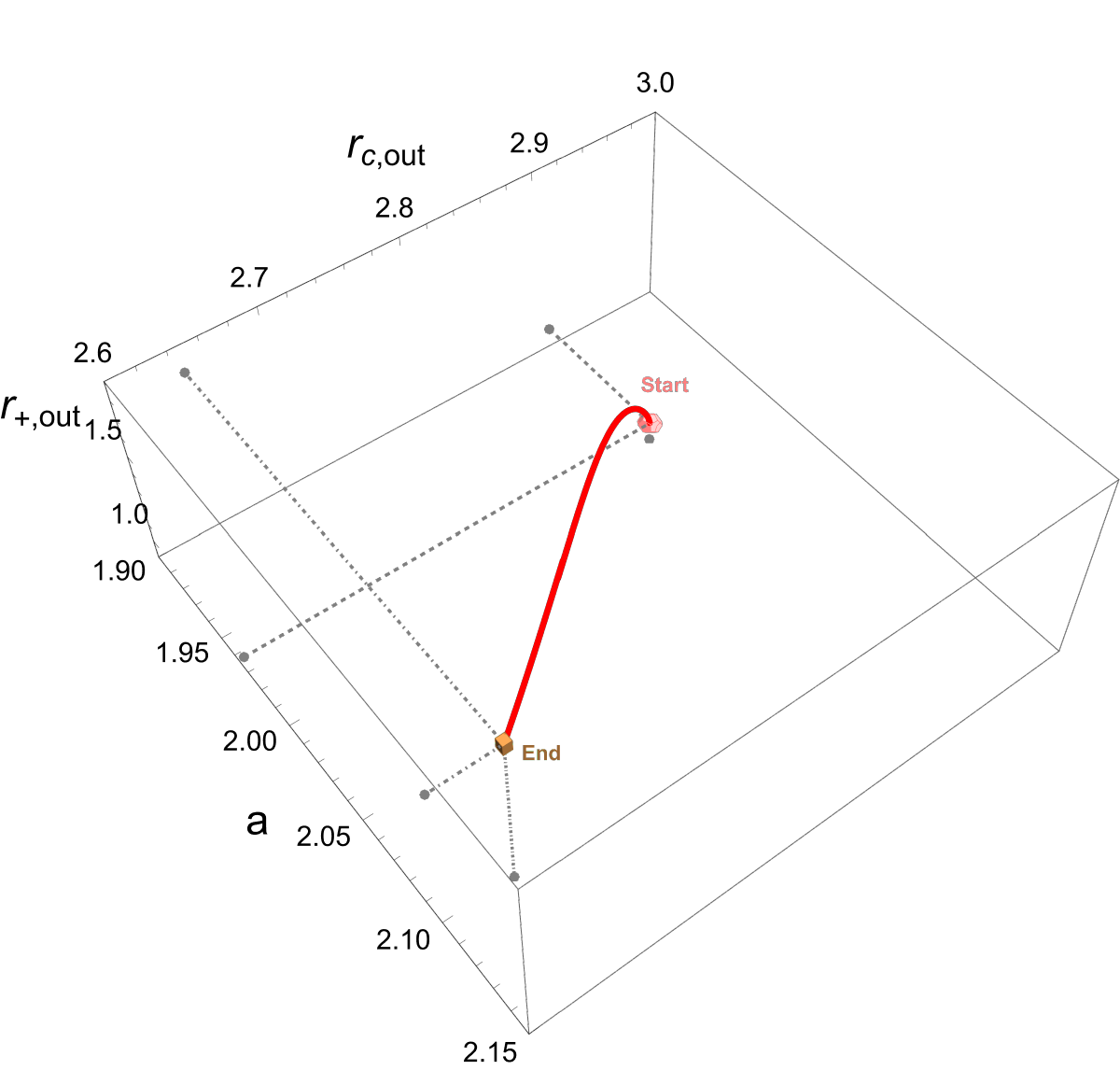}\label{EL_Plot}}\hfill
    \subfigure[]
    {\includegraphics[width=0.47\linewidth]{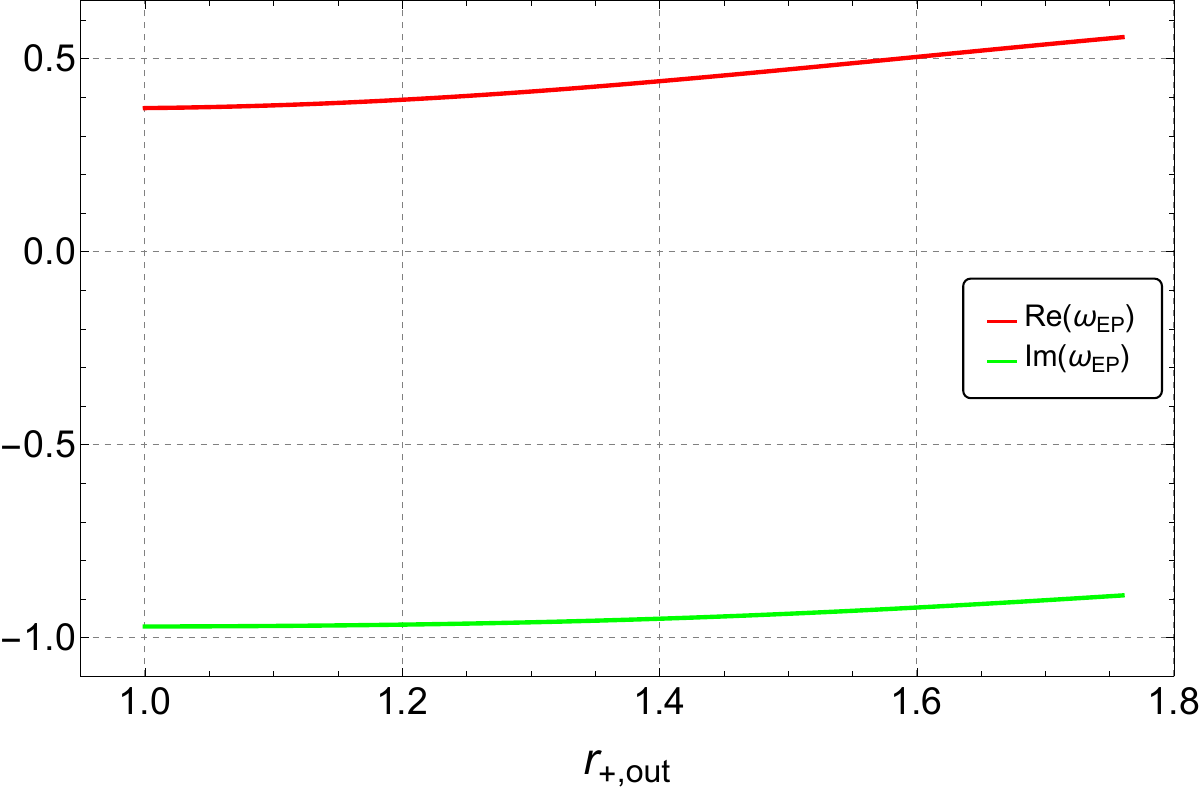}\label{omega_EL_rp_out_Plot}}
    \caption{Numerically traced exceptional line and the corresponding exceptional frequency for the modes $n=6$ and $n=7$ with $r_{-,\text{in}}=0.8$, $r_{\text{c,in}}=3$ and $\ell=2$. (a) Exceptional line in the three-dimensional parameter space $(a,r_{\text{c,out}},r_{+,\text{out}})$. The red curve represents the continuous locus of EPs, with the starting and ending points indicated explicitly. (b) Real and imaginary parts of the degenerate frequency $\omega_{\text{EP}}$ along the EL as functions of $r_{+,\text{out}}$ with $r_{+,\text{out}}\in[1,1.76]$.}
    \label{EL_Plot_and_omega_EL_rp_out}
\end{figure}

In~\cite{Cao:2025afs}, we have demonstrate that directional spectral sensitivity of the EL. In this study, we explain the reason theoretically. The directional spectral sensitivity for the EL is proposed in the followings. Assume that a point $\mathbf{p}_0$ is on the EL. Consider a small perturbation away from such exceptional point $\mathbf{p}_0$,
\begin{eqnarray}\label{eq:directional_perturbation}
    \mathbf{p}(\epsilon)=\mathbf{p}_0+\epsilon\widehat{\mathbf{u}}\, ,\qquad 0<\epsilon\ll1\, ,\qquad\big\|\widehat{\mathbf{u}}\big\|=1\, ,
\end{eqnarray}
where $\widehat{\mathbf{u}}$ is a prescribed unit direction in parameter space. Expanding $h$ around $\mathbf{p}_0$ gives
\begin{eqnarray}\label{eq:h_directional_expansion}
    h(\mathbf{p}_0+\epsilon\widehat{\mathbf{u}})=\epsilon\nabla h(\mathbf{p}_0)\cdot\widehat{\mathbf{u}}+\frac{\epsilon^2}{2}\widehat{\mathbf{u}}^{\mathrm{T}}\mathbf{H}_h(\mathbf{p}_0)\widehat{\mathbf{u}}+\mathcal{O}(\epsilon^3)\, ,
\end{eqnarray}
where the EP condition $h(\mathbf{p}_0)=0$ has been used, and $\mathbf{H}_h$ is the Hessian matrix of $h$. From Eq. (\ref{eq:h_F1_F2}), we have
\begin{eqnarray}\label{eq:mod_directional_h}
    \big|\nabla h(\mathbf{p}_0)\cdot\widehat{\mathbf{u}}\big|=\Big[\Big(\nabla F_1(\mathbf{p}_0)\cdot\widehat{\mathbf{u}}\Big)^2+\Big(\nabla F_2(\mathbf{p}_0)\cdot\widehat{\mathbf{u}}\Big)^2
    \Big]^{1/2}\, .
\end{eqnarray}
Substituting Eq. \eqref{eq:h_directional_expansion} into Eq. \eqref{eq:def_h}, we obtain, provided that $\nabla h(\mathbf{p}_0)\cdot\widehat{\mathbf{u}}\neq 0$,
\begin{eqnarray}\label{eq:square_root_splitting}
    \big|\omega_{+}(\mathbf{p}_0+\epsilon\widehat{\mathbf{u}})-\omega_{-}(\mathbf{p}_0+\epsilon\widehat{\mathbf{u}})
    \big|=
    \big|\epsilon
        \nabla h(\mathbf{p}_0)\cdot\widehat{\mathbf{u}}+\mathcal{O}(\epsilon^2)
    \big|^{1/2}=
    \sqrt{\epsilon}\big|\nabla h(\mathbf{p}_0)\cdot\widehat{\mathbf{u}}\big|^{1/2}+o(\sqrt{\epsilon})\, .
\end{eqnarray}
We therefore define the directional spectral sensitivity coefficient as
\begin{eqnarray}\label{eq:def_directional_sensitivity}
    C_{\widehat{\mathbf{u}}}:=\lim_{\epsilon\rightarrow 0^{+}}\frac{
        \big|\omega_{+}(\mathbf{p}_0+\epsilon\widehat{\mathbf{u}})-\omega_{-}(\mathbf{p}_0+\epsilon\widehat{\mathbf{u}})
        \big|
    }{\sqrt{\epsilon}}=\big|\nabla h(\mathbf{p}_0)\cdot\widehat{\mathbf{u}}\big|^{1/2}\, .
\end{eqnarray}
Define the real symmetric matrix which we call it the spectral sensitivity anisotropy matrix
\begin{eqnarray}\label{eq:def_K}
    \mathbf{K}:=\nabla F_1\nabla F_1^{\mathrm{T}}+\nabla F_2\nabla F_2^{\mathrm{T}}\, .
\end{eqnarray}
Note that the above notation is the dyadic notation. Therefore, from Eq. (\ref{eq:mod_directional_h}) and Eq. (\ref{eq:def_directional_sensitivity}), the directional spectral sensitivity can be written compactly as
\begin{eqnarray}\label{eq:C_matrix_form}
    C_{\widehat{\mathbf{u}}}=
    \Big(\widehat{\mathbf{u}}^{\mathrm{T}}
        \mathbf{K}\widehat{\mathbf{u}}\Big)^{1/4}\, .
\end{eqnarray}
Thus, the spectrum splitting near a second-order exceptional point takes the universal square-root form
\begin{eqnarray}\label{eq:universal_square_root}
    |\omega_{+}-\omega_{-}|=C_{\widehat{\mathbf{u}}}\sqrt{\epsilon}+o(\sqrt{\epsilon})\, .
\end{eqnarray}
However, for $\widehat{\mathbf{u}}=\mathbf{T}$, where $\mathbf{T}$ is unit tangent vector to the exceptional line, Eq. \eqref{eq:C_matrix_form} gives
\begin{eqnarray}\label{eq:tangent_sensitivity_zero}
        C_{\mathbf{T}}=0\, ,\qquad\text{with}\qquad \mathbf{T}=\frac{\nabla F_1\times\nabla F_2}{\big\|\nabla F_1\times\nabla F_2
        \big\|}\, .
\end{eqnarray} 
This means that matrix $\mathbf{K}$ has an eigenvalue of $0$, denoted as $\lambda_0$. More importantly, we can see that the first-order square-root splitting vanishes along the EL tangent. Therefore, in this case, we obtain $|\omega_{+}-\omega_{-}|=\mathcal{O}(\epsilon)$ rather than $\mathcal{O}(\sqrt{\epsilon})$. To demonstrate the reliability of our calculation results, we examine the results of two methods for calculating the tangent of the EL. The first method directly uses the result of the exceptional line to find the tangent, while the second method uses the eigenvectors corresponding to the zero eigenvalue of matrix $\mathbf{K}$, and finds that their results are highly consistent.

Different $\widehat{\mathbf{u}}$ determines different $C_{\widehat{\mathbf{u}}}$. It can be proved that the maximum sensitivity direction lies in the normal plane of the exceptional line. For more details about it, one can refer to Appendix \ref{thm: maximum_sensitivity_direction}. Furthermore, the maximum sensitivity over all unit directions is
\begin{eqnarray}\label{eq:C_max}
        C_{\max}=\lambda_{+}^{1/4}\, ,
\end{eqnarray}
where $\lambda_{+}$ is the maximum eigenvalue of matrix $\mathbf{K}$, and $0=\lambda_0<\lambda_{-}\leq\lambda_{+}$ are three eigenvalues of $\mathbf{K}$. The eigenvector corresponding to $\lambda_0$ is the tangent direction, whereas the eigenvectors associated with $\lambda_{-}$ and $\lambda_{+}$ lie in the normal plane. For the definition of directional sensitivity (\ref{eq:C_matrix_form}), it should only be compared within the normal plane of the exceptional line because one has $C_{\mathbf{T}}=0$. If the perturbation is restricted to the normal plane, the minimum and maximum sensitivities are instead
\begin{eqnarray}\label{eq:C_min_max_normal}
        C_{\min}^{(\perp)}=\lambda_{-}^{1/4}\, ,
        \qquad C_{\max}^{(\perp)}=\lambda_{+}^{1/4}\, .
\end{eqnarray}
The corresponding eigenvectors determine the least-sensitive and most-sensitive transverse perturbation directions, respectively. The expressions of $\lambda_{+}$ and $\lambda_{-}$ can be found in Appendix \ref{app: two_eigenvalues}. From Eqs. (\ref{eq:C_min_max_normal}), one can define the anisotropy of spectral sensitivity as follow
\begin{eqnarray}\label{anisotropy_of_spectral_sensitivity}
    \mathcal{A}:=\frac{C_{\max}^{(\perp)}}{C_{\min}^{(\perp)}}\, .
\end{eqnarray}
The anisotropy of spectral sensitivity $\mathcal{A}$ quantifies the directional dependence of the transverse spectral splitting near an EL.

\begin{figure}[htbp]
    \centering
    \subfigure[]
   {\includegraphics[width=0.47\linewidth]{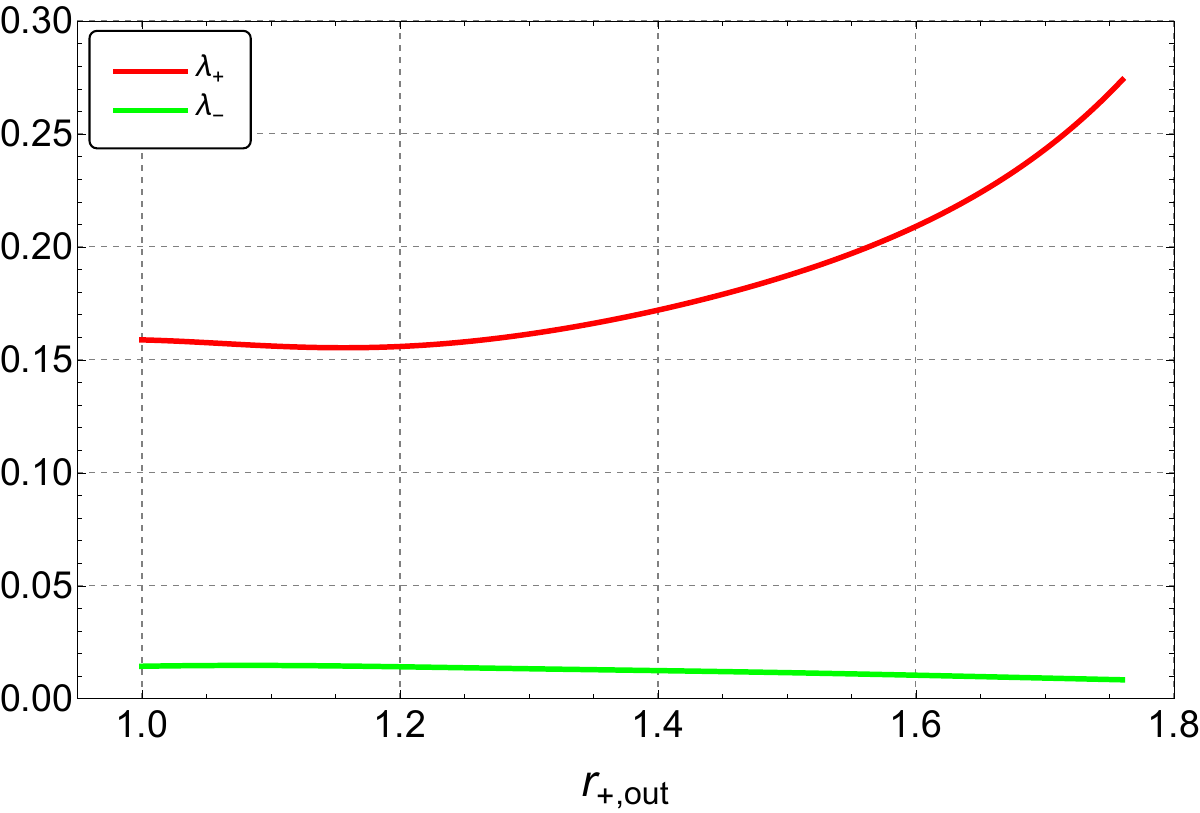}\label{lambda_p_lambda_m_Plot}}\hfill
    \subfigure[]
    {\includegraphics[width=0.47\linewidth]{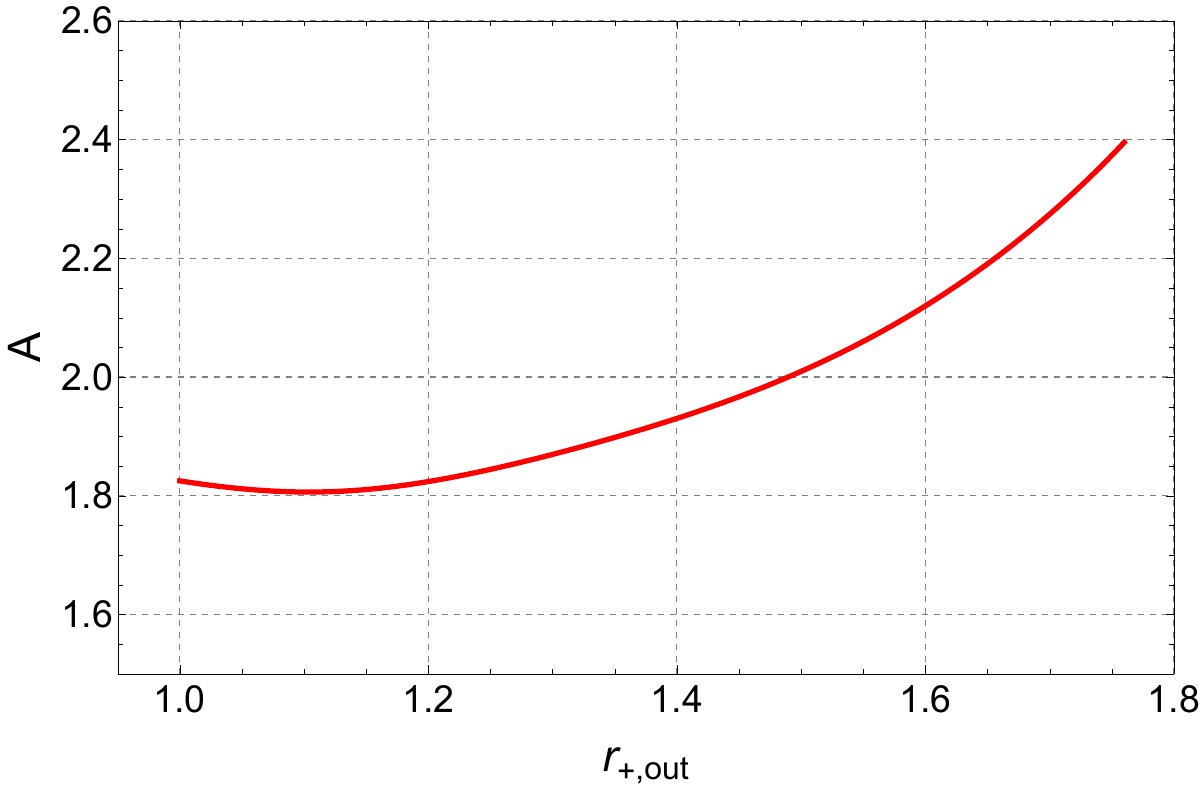}\label{A_Plot}}
    \caption{(a) Two eigenvalues $\lambda_{+}$ and $\lambda_{-}$ of the matrix $\mathbf{K}$ versus $r_{+,\text{out}}$, with $r_{+,\text{out}}\in[1,1.76]$. (b) The anisotropy of spectral sensitivity $\mathcal{A}=(\lambda_{+}/\lambda_{-})^{1/4}$. The results show that the spectral response remains anisotropic along the EL, with the directional sensitivity increasing gradually as $r_{+,\text{out}}$ increases.}
    \label{lambda_p_lambda_m_and_A}
\end{figure}

When $\mathcal{A}=1$, the spectral response is isotropic in the normal plane, meaning that perturbations in all transverse directions produce the same leading-order spectrum splitting. When $\mathcal{A}>1$, the spectral response is anisotropic, and the magnitude of the eigenfrequency splitting depends on the transverse perturbation direction. A large value, $\mathcal{A}\gg1$, indicates strong anisotropy. This means that the system is significantly more sensitive to perturbations along one principal normal direction than along the other. In the limiting case $\mathcal{A}\to\infty$, the leading-order square-root response vanishes along one principal direction in the normal plane, while remaining finite along the other. This behavior signals a loss of regularity of the local EL structure and may indicate the approach to a higher-order degeneracy, a branch reconstruction, or another singular geometric configuration. 

Fig. \ref{lambda_p_lambda_m_and_A} shows the evolution of the two eigenvalues, $\lambda_{+}$ and $\lambda_{-}$, together with the corresponding anisotropy measure $\mathcal{A}$ along the EL. Throughout the parameter range considered here, the two eigenvalues remain positive and satisfy $\lambda_{+} \gg \lambda_{-}$ indicating a pronounced directional dependence of the spectral sensitivity in the normal plane of the EL. The larger eigenvalue $\lambda_{+}$ first decreases slightly and reaches a shallow minimum around $r_{+,\text{out}}\simeq 1.15$, after which it increases continuously and becomes increasingly steep for larger $r_{+,\text{out}}$. In contrast, $\lambda_{-}$ varies much more slowly and exhibits an overall decreasing trend. Consequently, the separation between the two sensitivities becomes progressively larger as $r_{+,\mathrm{out}}$ increases. Fig. \ref{A_Plot} displays the anisotropy of spectral sensitivity (\ref{anisotropy_of_spectral_sensitivity}) which is the ratio between the maximum and minimum spectral sensitivity coefficients in the normal plane. One finds that $\mathcal{A}$ has a shallow minimum of about $1.8$ around $r_{+,\mathrm{out}}\simeq 1.1$, and then increases monotonically to about $2.4$ as $r_{+,\text{out}}$ becomes larger. This indicates that the spectral response is anisotropic over the EL, and that the directional contrast of the spectral sensitivity becomes progressively stronger toward larger $r_{+,\text{out}}$. 

\section{Exceptional line adapted parametrization of quasinormal modes}\label{sec: parameterization_EL}
Parametrized QNM formulations and tests provide a model-independent way to search for possible deviations of black hole spectra from the predictions of general relativity (GR)~\cite{Cardoso:2019mqo,McManus:2019ulj,Cano:2024jkd,Hirano:2024fgp,Kimura:2020mrh,Yu:2025wpb,Volkel:2026qqz,Chen:2026ack}. These conventional parametrizations rely, explicitly or implicitly, on the existence of a regular perturbative description of the QNM spectrum in the neighborhood of the reference configuration. For a QNM branch depending smoothly on a set of physical parameters $\mathbf{p}=(p_1,p_2,\ldots)$, one may therefore locally write
\begin{eqnarray}\label{eq:ordinary_param_QNM}
    \omega_{\ell mn}(\mathbf{p}+\delta\mathbf{p})=\omega_{\ell mn}(\mathbf{p})+\delta\omega_{\ell mn}(\mathbf{p})+\mathcal{O}(\|\delta\mathbf{p}\|^2)\, ,
\end{eqnarray}
where the deviation is assumed to depend analytically on the parameters $\mathbf{p}$, which means that
\begin{eqnarray}\label{eq:linear_QNM_deviation}
    \delta\omega_{\ell mn}(\mathbf{p})=\sum_a\frac{\partial\omega_{\ell mn}}{\partial p_a}\delta p_a\, .
\end{eqnarray}

Equation \eqref{eq:linear_QNM_deviation} assumes that the QNM spectrum is an analytic and single-valued function of the physical parameters. However, this assumption breaks down in the vicinity of an EL. From Eq. (\ref{sum_omega}) and Eq. (\ref{squared_difference_omega}), two spectra are obtained as
\begin{eqnarray}\label{eq:EP_QNM_general}
    \omega_\pm(\mathbf p)=\overline{\omega}(\mathbf p)\pm\frac{1}{2}\sqrt{h(\mathbf p)}\, ,\qquad\text{with}\qquad \overline{\omega}(\mathbf p)\equiv\frac{\omega_{+}(\mathbf p)+\omega_{-}(\mathbf p)}{2}\, .
\end{eqnarray}
Using Eq. \eqref{eq:EP_QNM_general}, one obtains
\begin{eqnarray}\label{eq:sensitivity_EL}
    \frac{\partial\omega_\pm}{\partial p_a}=\frac{\partial\overline{\omega}}{\partial p_a}\pm\frac{1}{4\sqrt{h}}\frac{\partial h}{\partial p_a}\, .
\end{eqnarray}
Near the exceptional line, $h\rightarrow0$ with $\partial h/\partial p_a$ being finite and hence
\begin{eqnarray}\label{eq:sensitivity_divergence}
    \Big|\frac{\partial\omega_\pm}{\partial p_a}\Big|\rightarrow\infty\, .
\end{eqnarray}
Therefore, the linear coefficients used in conventional parameterized QNM models become singular and cannot provide a proper description near an EL. This regular parametrization requires modification in the vicinity of an EL. 

The directional spectral sensitivity derived in Sec. \ref{sec: ELs} suggests a natural parametrization adapted to the local geometry of the EL. Since a displacement tangent to the EL and a displacement transverse to it have qualitatively different spectral consequences, it is advantageous to replace the original physical parameter coordinates by local coordinates constructed from the Frenet frame of the EL. Using the Frenet frame, any nearby point of EL in the parameter space can be decomposed as
\begin{eqnarray}
    \delta\mathbf p=q_t\mathbf T+q_n\mathbf N+q_b\mathbf B \, ,
\end{eqnarray}
where $\mathbf N$ and $\mathbf B$ are principal normal vector and auxiliary normal vector. The tangential displacement $q_t$ only changes the position along the exceptional line and therefore can be absorbed into the parameter $s$. Hence the local coordinates around the EL are chosen as
\begin{eqnarray}\label{eq:tube_coordinate}
    \mathbf p(s,q_n,q_b)=\mathbf p_{\text{EL}}(s)+q_n\mathbf N(s)+q_b\mathbf B(s)\, .
\end{eqnarray}
Here, $q_n$ and $q_b$ represent the two independent perturbation directions perpendicular to the EL. It is important to distinguish the roles of these coordinates. The coordinate $s$ labels different exceptional points along the EL and determines the evolution of the degenerate frequency $\omega_{\text{EP}}(s)$, whereas $(q_n,q_b)$ describe physical departures away from the degeneracy and are responsible for the splitting of the two QNM branches.

On the one hand, the average spectrum $\overline{\omega}$ is expanded around the EL,
\begin{eqnarray}\label{eq:omega_average_expansion}
    \overline{\omega}(s,q_n,q_b)=\omega_{\text{EP}}(s)+v_n(s)q_n+v_b(s)q_b+\mathcal{O}(q^2)\, ,
\end{eqnarray}
where 
\begin{eqnarray}
    v_n(s)=\frac{\partial\overline{\omega}}{\partial q_n}\Big|_{\text{EL}}\, ,\qquad v_b(s)=\frac{\partial\overline{\omega}}{\partial q_b}\Big|_{\text{EL}}\, .
\end{eqnarray}
On the other hand, the Taylor expansion of $h$ gives
\begin{eqnarray}\label{eq:h_expansion}
    h(s,q_n,q_b)=\alpha_n(s)q_n+\alpha_b(s)q_b+\frac12\beta_{nn}(s)q_n^2+\beta_{nb}(s)q_nq_b+\frac12\beta_{bb}(s)q_b^2+\cdots\, ,
\end{eqnarray}
in which the coefficients are
\begin{eqnarray}
    \alpha_n(s)=\frac{\partial h}{\partial q_n}\Big|_{\text{EL}}\, ,\quad  \alpha_b(s)=\frac{\partial h}{\partial q_b}\Big|_{\text{EL}}\, ,\quad \beta_{nn}=\frac{\partial^2 h}{\partial q_n^2}\Big|_{\text{EL}}\, ,\quad \beta_{nb}=\frac{\partial^2 h}{\partial q_n\partial q_b}\Big|_{\text{EL}}\, ,\quad \beta_{bb}=\frac{\partial^2 h}{\partial q_b^2}\Big|_{\text{EL}}\, .
\end{eqnarray}
Substituting Eqs. \eqref{eq:omega_average_expansion} and \eqref{eq:h_expansion} into Eq. \eqref{eq:EP_QNM_general}, the QNM spectra near the EL become
\begin{eqnarray}
    \omega_\pm(s,q_n,q_b)=\omega_{\text{EP}}(s)+v_n(s)q_n+v_b(s)q_b\pm\frac12\sqrt{\alpha_n(s)q_n+\alpha_b(s)q_b+\frac12\beta_{nn}(s)q_n^2+\beta_{nb}(s)q_nq_b+\frac12\beta_{bb}(s)q_b^2}\, ,
\end{eqnarray}
which is the generalized parameterized QNM model around an EL. If one is sufficiently close to the exceptional line: $q_n, q_b \ll 1$, neglecting second-order terms: $\beta_{ij} q_i q_j \approx 0$. Keeping only the leading-order terms gives
\begin{eqnarray}\label{eq:leading_EL_QNM}
    \omega_\pm(s,q_n,q_b)\simeq\omega_{\text{EP}}(s)+v_nq_n+v_bq_b\pm\frac12\sqrt{\alpha_nq_n+\alpha_bq_b}\, .
\end{eqnarray}
This is the parameterized QNMs formulation near the exceptional line. Note that $v_n$, $v_b$, $\alpha_n$ and $\alpha_b$ are viewed as the expansion bases.


\section{conclusions and discussion}\label{conclusions}
In this work, we have investigated the QNM spectrum and exceptional structures of the Reissner-Nordstr\"{o}m-de Sitter black hole surrounded by a static thin shell of matter. For simplicity, we have considered a conformal scalar field which does not interact directly with the shell. The interior and exterior geometries are both RN-dS spacetimes, while the shell provides an interface connecting the two regions and induces a relative rescaling of the time coordinates. Consequently, the frequencies measured in the two regions satisfy $\omega_{\rm out}=\eta(a)\omega_{\rm in}$. By imposing the matching conditions of the scalar field across the shell, together with the ingoing boundary condition at the event horizon and the outgoing boundary condition at the cosmological horizon, we obtained the QNM condition for the composite spacetime.

We first studied the migration of the QNM spectrum under variations of the shell radius and the exterior geometric parameters. The numerical results show that the higher overtones are considerably more sensitive to these variations than the low overtones. In particular, the $n=6$ and $n=7$ modes exhibit a mode exchange when the parameters $(a,r_{\text{c,out}})$ are varied. By encircling the corresponding parameter region, the two QNM branches are permuted after one cycle, which demonstrates the existence of an exceptional point. Combining the rectangular subdivision procedure with the nine-point method, we locate this EP. By promoting $r_{+,\text{out}}$ to an additional control parameter, the two-dimensional parameter space is extended to the three-dimensional space $(a,r_{\text{c,out}},r_{+,\text{out}})$. In this enlarged parameter space. Numerically, we obtain an exceptional line. Along the EL, both the location of the degeneracy in parameter space and the corresponding exceptional frequency vary smoothly. 

We further developed a local description of the directional spectral sensitivity around the EL. For a perturbation $\mathbf p=\mathbf p_0+\epsilon\widehat{\mathbf u}$ away from a point $\mathbf p_0$ on the EL, the splitting of the two QNM branches takes the universal square-root form $|\omega_+-\omega_-|=C_{\widehat{\mathbf u}}\sqrt{\epsilon}+o(\sqrt{\epsilon})$, where $C_{\widehat{\mathbf u}}=(\widehat{\mathbf u}^{T}\mathbf{K}\widehat{\mathbf u})^{1/4}$, and $\mathbf{K}=\nabla F_1\nabla F_1^{\mathrm T}+\nabla F_2\nabla F_2^{\mathrm T}$. The tangent vector of the EL belongs to the null space of $\mathbf{K}$, and therefore the leading square-root splitting vanishes identically for a perturbation tangent to the EL. In contrast, generic transverse perturbations produce the characteristic $\sqrt{\epsilon}$ splitting. We have also shown that the direction of maximum spectral sensitivity necessarily lies in the normal plane of the EL. The two nonzero eigenvalues $\lambda_+$ and $\lambda_-$ of $\mathbf{K}$ determine the maximum and minimum transverse sensitivities, respectively, $C_{\max}^{(\perp)}=\lambda_+^{1/4}$, $C_{\min}^{(\perp)}=\lambda_-^{1/4}$. This naturally leads to the anisotropy factor ${\mathcal A}=(\lambda_{+}/\lambda_{-})^{1/4}$. Our numerical results give ${\cal A}>1$ throughout the EL considered here. It reaches a shallow minimum of approximately $1.8$ and subsequently grows to about $2.3$, showing that the transverse spectral response is intrinsically anisotropic and that the anisotropy becomes stronger toward larger $r_{+,\text{out}}$.

Finally, we discussed the implication of the EL structure for parametrized QNM frameworks. Conventional linear parametrizations assume an analytic and single-valued dependence of a QNM spectrum on the physical parameters. Near an EL, however, $\omega_\pm=\bar{\omega}\pm\sqrt{h}/2$, so that the derivatives of the individual QNM branches generally diverge as $h\rightarrow0$. A regular Taylor expansion of each branch therefore ceases to be appropriate. Motivated by the local geometry of the EL, we introduced coordinates adapted to its Frenet frame, $\mathbf p(s,q_n,q_b)=\mathbf p_{\rm EL}(s)+q_n\mathbf N(s)+q_b\mathbf B(s)$, and obtained an EL-adapted parametrization of the QNM spectrum. To leading order, it takes the form (\ref{eq:leading_EL_QNM}). This parametrization explicitly incorporates both the geometry of the exceptional line and the nonanalytic square-root splitting of the QNM spectrum.

Our results demonstrate that exceptional lines provide a natural framework for organizing the strongly direction-dependent spectral response of black hole QNMs in multidimensional parameter spaces. They also show that the local geometry of exceptional structures has direct consequences for parametrized descriptions of black hole spectra. It would be interesting to extend the present analysis to gravitational and electromagnetic perturbations, more general thin-shell matter models (including the coupling situation), and rotating black holes. 


\section*{Acknowledgement}
This work is supported by the National Natural Science Foundation of China with grant No. 12505067. This work is also supported by the National Key R\&D Program of China grant No. 2022YFC2204603, by the National Natural Science Foundation of China with grants No. 12475063, No. 12075232.

\appendix
\section{Israel junction conditions}\label{app: Israel junction conditions}
A static thin shell denote by $\Sigma$ is placed at $r=a$, dividing spacetime into an interior region ($r < a$, metric function $f_{\text{in}}(r)$) and an exterior region ($r > a$, metric function $f_{\text{out}}(r)$). The induced metric denoted by $h_{ab}$ on the shell is written as
\begin{eqnarray}\label{induced_metric}
    \mathrm{d}s_\Sigma^2 = -\mathrm{d}\tau^2 + a^2 \mathrm{d}\Omega^2\, ,
\end{eqnarray}
where $\tau$ is the proper time on the shell, related to the coordinate time by $\mathrm{d}\tau=\sqrt{f(a)}\mathrm{d}t$. The unit normal vectors are
\begin{eqnarray}\label{normal_vectors}
    n^\mu=\Big(0,\sqrt{f(a)},0,0\Big)\, ,\quad n_\mu=\Big(0,1/\sqrt{f(a)},0,0\Big)\, ,
\end{eqnarray}
with $n^\mu n_\mu=1$. For both inside and outside spacetimes, the vector points at the direction of increasing $r$ coordinate ($\sqrt{f(a)}>0$). From Eq. (\ref{induced_metric}) and Eq. (\ref{normal_vectors}), we can obtain the extrinsic curvature $K_{ij}$ in which
\begin{eqnarray}\label{extrinsic_curvature}
    K_{\tau\tau}=-\frac{f^{\prime}(a)}{2\sqrt{f(a)}}\, ,\quad K_{\theta\theta}=a\sqrt{f(a)}\, ,\quad K_{\phi\phi}=a\sqrt{f(a)}\sin^2\theta\, ,\quad \text{and}\quad K=\frac{f^{\prime}(a)}{2\sqrt{f(a)}}+\frac{2\sqrt{f(a)}}{a}\, .
\end{eqnarray}
 
The Israel junction conditions, which relate the discontinuities in the extrinsic curvature $K_{ij}$ over the boundary of a thin surface of matter to its surface stress-energy tensor $S_{ij}$, can be expressed as~\cite{Israel:1966rt}
\begin{eqnarray}\label{Israel_junction_conditions}
    [[K_{ij}]]-h_{ij}[[K]]=-8\pi S_{ij}\, ,
\end{eqnarray}
where the discontinuities over the shell defined as
\begin{eqnarray}\label{discontinuities}
    [[A]]\equiv \Big(\lim_{r\to a^{+}}A\Big)-\Big(\lim_{r\to a^{-}}A\Big)\, .
\end{eqnarray}
From Eqs. (\ref{extrinsic_curvature}) and Eq. (\ref{Israel_junction_conditions}), we have two connection conditions as follows
\begin{eqnarray}
    \frac{2}{a}\Big[\Big[\sqrt{f(r)}\Big]\Big]=-8\pi S_{\tau\tau}\, ,\quad\text{and}\quad \frac{a^2}{2}\Big[\Big[\frac{f^{\prime}(r)}{\sqrt{f(r)}}\Big]\Big]+a\Big[\Big[\sqrt{f(r)}\Big]\Big]=8\pi S_{\theta\theta}\, .
\end{eqnarray}
For a perfect fluid,
$S_{ij}$ can be written as
\begin{eqnarray}
    S_{ij}=\text{diag}(\sigma,pa^2,pa^2\sin^2\theta)\, ,
\end{eqnarray}
where $\sigma$ is the surface energy density and $p$ is the surface tension.

\section{The Heun's equation}\label{app: Heun_equation}
In the appendix, we will solve Eq. (\ref{master_equation_frequency_domain}) by using the Heun's function method. It is not difficult to see that the differential equation (\ref{master_equation_frequency_domain}) has five regular singular points at $r=r_n$, $r=0$, $r=r_{-}$, $r=r_{+}$ and $r=r_\text{c}$. Note that the point $r=\infty$ is the regular point. In terms of $r$, Eq. (\ref{master_equation_frequency_domain}) becomes a standard form
\begin{eqnarray}\label{standard_equation}
    \frac{\mathrm{d}^2\Psi}{\mathrm{d}r^2}+p(r)\frac{\mathrm{d}\Psi}{\mathrm{d}r}+q(r)\Psi=0\, ,
\end{eqnarray}
where the functions $p(r)$ and $q(r)$ read
\begin{eqnarray}
    p(r)=\frac{f^{\prime}(r)}{f(r)}\, ,\quad q(r)=\frac{\omega^2-V(r)}{f^2(r)}\, . 
\end{eqnarray}
Here, $\prime$ is the derivative with respect to $r$. It is known that general Heun's equation has four regular singular points including the infinity. However, this equation has five regular singular points. Therefore, it is necessary for us to eliminate a singular point. Moreover, not all regular singular points can be eliminated through analytical transformations. We can use the Liouville invariant to discriminate this. The Liouville invariant is defined as
\begin{eqnarray}\label{Liouville_invariant}
    I(r)\equiv q(r)-\frac12 p^{\prime}(r)-\frac14\Big[p(r)\Big]^2\, .
\end{eqnarray}
If $I(r)$ is analytical at some regular singular point, then one can do a transformation to eliminate such regular singular point. After some calculation, we find that although $r=0$ is a regular singular point, the Liouville invariant is analytical at $r=0$. We consider the M\"{o}bius transformation as follow
\begin{eqnarray}\label{Mobius_transformation}
    z(r)=\frac{(r_\text{c}-r_{-})(r-r_{+})}{(r_\text{c}-r_{+})(r-r_{-})}\, .
\end{eqnarray}
In this transformation, $(r_{\text{c}},r_{+},r_{-},r_n,0)$ are mapped into $(1,0,\infty,z_n,z_0)$, respectively, where
\begin{eqnarray}\label{z_n_and_z_0}
    z_n=\frac{(r_{\text{c}}-r_{-})(r_n-r_{+})}{(r_{\text{c}}-r_{+})(r_n-r_{-})}>1\, ,\quad \text{and}\quad z_0=\frac{(r_{\text{c}}-r_{-})r_{+}}{(r_\text{c}-r_{+})r_{-}}\, .
\end{eqnarray}
We then further take the following transformation for $\Psi$, namely
\begin{eqnarray}
    \Psi(z)=z^{\rho_{+,1}}(z-1)^{\rho_{\text{c},1}}(z-z_n)^{\rho_{n,1}}(z-z_0)u(z)\, ,
\end{eqnarray}
where the scaling factor $z-z_0$ is used to eliminate the regular singular point $z=z_0$. Indices $\rho_{+,1}$, $\rho_{\text{c},1}$ and $\rho_{n,1}$ are shown in Tab. \ref{indices}. Note that similar operations also occur in the cases of Kerr-de Sitter black holes~\cite{Hatsuda:2020sbn,Oshita:2021iyn} and global monopole Reissner-Nordstr\"{o}m-de Sitter black hole~\cite{Li:2026zsg}. The new function $u(z)$ will satisfy the Heun's equation
\begin{eqnarray}\label{Heun_equation}
    \frac{\mathrm{d}^2u(z)}{\mathrm{d}z^2}+\Big(\frac{\gamma}{z}+\frac{\delta}{z-1}+\frac{\epsilon}{z-z_n}\Big)\frac{\mathrm{d}u(z)}{\mathrm{d}z}+\frac{(\alpha\beta z-q)u(z)}{z(z-1)(z-z_n)}=0\, ,
\end{eqnarray}
where $\gamma+\delta+\epsilon=\alpha+\beta+1$. We give six parameters in the above Heun's equation, where the condition $r_n=-(r_{+}+r_{-}+r_{\text{c}})$ has been used. Note that $\alpha$ and $\beta$ can be exchanged. To increase readability, we explicitly provide expressions of six parameters $\gamma$, $\delta$, and $\epsilon$, $\alpha$, $\beta$ and $q$,
\begin{eqnarray}\label{gamma}
    \gamma(\omega)=\frac{\Lambda  r_{\text{c}}^2 (r_{+}-r_{-})+\Lambda  r_{\text{c}} (r_{+}^2-r_{-}^2)+r_{+} (\Lambda r_{-}^2+\Lambda  r_{-}r_{+}-2 \Lambda r_{+}^2+6 \mathrm{i} r_{+} \omega )}{\Lambda  (r_{\text{c}}-r_{+}) (r_{+}-r_{-}) (r_{\text{c}}+r_{-}+2 r_{+})}\, ,
\end{eqnarray}
\begin{eqnarray}\label{delta}
    \delta(\omega)=\frac{2\Lambda r_{\text{c}}^3-r_{\text{c}}^2 (\Lambda r_{-}+\Lambda  r_{+}+6 \mathrm{i}\omega)-\Lambda r_{\text{c}}(r_{-}^2+r_{+}^2)+\Lambda  r_{-} r_{+} (r_{-}+r_{+})}{\Lambda (r_{\text{c}}-r_{-}) (r_{\text{c}}-r_{+})(2r_{\text{c}}+r_{-}+r_{+})}\, ,
\end{eqnarray}
\begin{eqnarray}\label{epsilon}
    \epsilon(\omega)&=&\Bigg\{\Lambda \Big[2 r_{\text{c}}^3+7 r_{\text{c}}^2 (r_{-}+r_{+})+r_{\text{c}}\Big(7 r_{-}^2+16 r_{-} r_{+}+7 r_{+}^2\Big)+2 r_{-}^3+7 r_{-}^2 r_{+}+7r_{-} r_{+}^2+2 r_{+}^3\Big]\Bigg\}^{-1}\nonumber\\
    &&\times\Bigg\{2 \Lambda  r_{\text{c}}^3+r_{\text{c}}^2 \Big(7 \Lambda  r_{-}+7 \Lambda r_{+}+6 \mathrm{i} \omega \Big)+r_{\text{c}} \Big[7 \Lambda r_{-}^2+4 r_{-}(4 \Lambda  r_{+}+3 \mathrm{i} \omega )+r_{+}(7 \Lambda r_{+}+12 \mathrm{i} \omega )\Big]\nonumber\\
    &&+(r_{-}+r_{+}) \Big(2 \Lambda r_{-}^2+5 \Lambda  r_{-} r_{+}+6 \mathrm{i} r_{-} \omega +2 \Lambda r_{+}^2+6 \mathrm{i} r_{+}\omega \Big)\Bigg\}\, ,
\end{eqnarray}
\begin{eqnarray}\label{q}
    q(\omega)=\frac{-3 \ell^2 (r_{-}-r_{+})-3 \ell (r_{-}-r_{+})+\Lambda  r_{\text{c}}^2 (r_{+}-r_{-})+\Lambda  r_{\text{c}} (r_{+}^2-r_{-}^2)+r_{-} r_{+}(\Lambda r_{-}-\Lambda r_{+}+6 \mathrm{i} \omega )}{\Lambda  (r_{\text{c}}-r_{+}) (r_{+}-r_{-}) (r_{\text{c}}+2 r_{-}+r_{+})}\, ,
\end{eqnarray}
and
\begin{eqnarray}\label{alpha_and_beta}
    \alpha=1\, ,\quad \beta(\omega)=\frac{\Lambda  r_{\text{c}}^2 (r_{-}-r_{+})+\Lambda r_{\text{c}}(r_{-}^2-r_{+}^2)+r_{-}(-2 \Lambda r_{-}^2+\Lambda  r_{-} r_{+}-6 \mathrm{i} r_{-} \omega +\Lambda r_{+}^2)}{\Lambda  (r_{\text{c}}-r_{-}) (r_{-}-r_{+}) (r_{\text{c}}+2 r_{-}+r_{+})}\, ,
\end{eqnarray}
where $\Lambda$ is given by Eq. (\ref{cosmological_constant}). As $r_{-}\to0$, RN-dS black hole becomes SdS black hole and the above six parameters reduce to the corresponding parameters in SdS black hole case~\cite{Wu:2025wbp}.

\begin{table}[t]
\setlength{\tabcolsep}{12pt} 
 \centering
\begin{tabular}{c c c c c}
 \hline\hline
Regular singular points~~~~~~& $a_0$ ~~~~~~& $b_0$ ~~~~~~&Indices $(\rho_1,\rho_2)$ ~~~~~~&  \\
 \hline\noalign{\vskip 5pt}
$r=r_{n}$~~~~~~& $1$ ~~~~~~& $\displaystyle \frac{\omega^2}{\Big[f^{\prime}(r_{n})\Big]^2}$ ~~~~~~&$\displaystyle \Bigg(+\frac{\mathrm{i}\omega}{f^{\prime}(r_n)}\, ,-\frac{\mathrm{i}\omega}{f^{\prime}(r_n)}\Bigg)$ ~~~~~~&  \\\noalign{\vskip 5pt}
\hline\noalign{\vskip 5pt}
$r=0$~~~~~~& $-2$ ~~~~~~& $2$ ~~~~~~&$\displaystyle (2, 1)$ ~~~~~~&  \\\noalign{\vskip 5pt}
\hline\noalign{\vskip 5pt}
$r=r_{-}$~~~~~~& $1$ ~~~~~~& $\displaystyle \frac{\omega^2}{\Big[f^{\prime}(r_{-})\Big]^2}$ ~~~~~~&$\displaystyle \Bigg(+\frac{\mathrm{i}\omega}{f^{\prime}(r_{-})}\, ,-\frac{\mathrm{i}\omega}{f^{\prime}(r_{-})}\Bigg)$ ~~~~~~&  \\\noalign{\vskip 5pt}
\hline\noalign{\vskip 5pt}
$r=r_{+}$~~~~~~& $1$ ~~~~~~& $\displaystyle \frac{\omega^2}{\Big[f^{\prime}(r_{+})\Big]^2}$ ~~~~~~&$\displaystyle \Bigg(+\frac{\mathrm{i}\omega}{f^{\prime}(r_+)}\, ,-\frac{\mathrm{i}\omega}{f^{\prime}(r_+)}\Bigg)$ ~~~~~~&  \\\noalign{\vskip 5pt}
\hline\noalign{\vskip 5pt}
$r=r_{\text{c}}$~~~~~~& $1$ ~~~~~~& $\displaystyle \frac{\omega^2}{\Big[f^{\prime}(r_{\text{c}})\Big]^2}$ ~~~~~~&$\displaystyle \Bigg(+\frac{\mathrm{i}\omega}{f^{\prime}(r_{\text{c}})}\, ,-\frac{\mathrm{i}\omega}{f^{\prime}(r_{\text{c}})}\Bigg)$ ~~~~~~&  \\\noalign{\vskip 5pt}
\hline\hline
\end{tabular}
\caption{Regular singular points of Eq. (\ref{standard_equation}), the corresponding coefficients $a_0$, $b_0$ in the index equation, $\rho(\rho-1)+a_0\rho+b_0=0$, and the associated indices $(\rho_1,\rho_2)$ for conformal scalar perturbation. The sum of all indices are $3$, which is satisfied with Fuchs condition.}
\label{indices}
\end{table}

Considering the QNM boundary conditions, which state that the solution is ingoing at the event horizon and the solution is outgoing at the cosmological horizon, one can find a solution $\Psi_{\text{in}}(\omega_{\text{in}},r)$ that satisfies ingoing condition at event horizon. For $r_{+,\text{in}}<r<a$,
\begin{eqnarray}\label{Psi_in}
    \Psi_{\text{in}}(\omega_{\text{in}},r)&=&(z_{\text{in}})^{\rho_{+,1,\text{in}}}(z_{\text{in}}-1)^{\rho_{\text{c},1,\text{in}}}(z_{\text{in}}-z_{n,\text{in}})^{\rho_{n,1,\text{in}}}(z_{\text{in}}-z_{0,\text{in}})\nonumber\\
    &&\times z_{\text{in}}^{1-\gamma_{\text{in}}}Hl(z_{n,\text{in}},(z_{n,\text{in}}\delta_{\text{in}}+\epsilon_{\text{in}})(1-\gamma_{\text{in}})+q_{\text{in}};\alpha_{\text{in}}+1-\gamma_{\text{in}},\beta_{\text{in}}+1-\gamma_{\text{in}},2-\gamma_{\text{in}},\delta_{\text{in}};z_{\text{in}})
\end{eqnarray}
with
\begin{eqnarray}\label{z_in}
        z_{\text{in}}=\frac{(r_\text{c,in}-r_{-,\text{in}})(r-r_{+,\text{in}})}{(r_\text{c,in}-r_{+,\text{in}})(r-r_{-,\text{in}})}\, .
\end{eqnarray}
For $a<r<r_{\text{c,out}}$, a solution $\Psi_{\text{out}}(\omega_{\text{out}},r)$ that satisfies outgoing condition at cosmological horizon, which is given by
\begin{eqnarray}\label{Psi_out}
    \Psi_{\text{out}}(\omega_{\text{out}},r)&=&(z_{\text{out}})^{\rho_{+,1,\text{out}}}(z_{\text{out}}-1)^{\rho_{\text{c},1,\text{out}}}(z_{\text{out}}-z_{n,\text{out}})^{\rho_{n,1,\text{out}}}(z_{\text{out}}-z_{0,\text{out}})\nonumber\\
    &&\times Hl(1-z_{n,\text{out}},\alpha_{\text{out}}\beta_{\text{out}}-q_{\text{out}};\alpha_{\text{out}},\beta_{\text{out}},\delta_{\text{out}},\gamma_{\text{out}};1-z_{\text{out}})
\end{eqnarray}
with
\begin{eqnarray}\label{z_out}
    z_{\text{out}}=\frac{(r_\text{c,\text{out}}-r_{-,\text{out}})(r-r_{+,\text{out}})}{(r_\text{c,out}-r_{+,\text{out}})(r-r_{-,\text{out}})}\, .
\end{eqnarray}

\section{Maximum spectral sensitivity direction}
\label{thm: maximum_sensitivity_direction}
\begin{theorem}[Maximum spectral sensitivity direction]
Let the exceptional line $\mathcal{L}_{\mathrm{EP}}$ in a three-dimensional parameter space be locally defined by
\begin{eqnarray}
    F_1(\mathbf{p})=0\, ,
    \qquad
    F_2(\mathbf{p})=0\, ,
\end{eqnarray}
where $F_1,F_2:\mathbb{R}^{3}\rightarrow\mathbb{R}$ are smooth functions. Suppose that, at a regular point $\mathbf{p}_0\in\mathcal{L}_{\mathrm{EP}}$, the gradients $\nabla F_1(\mathbf{p}_0)$ and $\nabla F_2(\mathbf{p}_0)$ are linearly independent. Define
\begin{eqnarray}\label{eq:K_theorem}
    \mathbf{K}=\nabla F_1\nabla F_1^{\mathrm{T}}+\nabla F_2\nabla F_2^{\mathrm{T}}\, ,
\end{eqnarray}
where all gradients are evaluated at $\mathbf{p}_0$, and define the directional spectral sensitivity coefficient by
\begin{eqnarray}\label{eq:C_theorem}
    C_{\widehat{\mathbf{u}}}=\Big(\widehat{\mathbf{u}}^{\mathrm{T}}\mathbf{K}\widehat{\mathbf{u}}
    \Big)^{1/4}\, ,
    \qquad
    \|\widehat{\mathbf{u}}\|=1\, .
\end{eqnarray}
Then every direction maximizing $C_{\widehat{\mathbf{u}}}$ is orthogonal to the tangent vector $\mathbf{T}$ of the exceptional line, namely,
\begin{eqnarray}
    \widehat{\mathbf{u}}_{\max}\cdot\mathbf{T}=0\, .
\end{eqnarray}
Therefore, the maximum-sensitivity direction lies in the normal plane of the exceptional line. Moreover, if $\lambda_{+}$ denotes the largest eigenvalue of $\mathbf{K}$, then
\begin{eqnarray}
    C_{\max}=\lambda_{+}^{1/4}\, ,
\end{eqnarray}
and $\widehat{\mathbf{u}}_{\max}$ is a normalized eigenvector of $\mathbf{K}$ corresponding to $\lambda_{+}$.
\end{theorem}

\begin{proof}
Since $\nabla F_1$ and $\nabla F_2$ are linearly independent, the unit tangent vector of the exceptional line is
\begin{eqnarray}
    \mathbf{T}=\frac{\nabla F_1\times\nabla F_2}{\|\nabla F_1\times\nabla F_2\|}\, .
\end{eqnarray}
It follows that
\begin{eqnarray}\label{eq:grad_T_zero}
    \nabla F_1\cdot\mathbf{T}=0\, ,
    \qquad
    \nabla F_2\cdot\mathbf{T}=0\, .
\end{eqnarray}
Using the definition of $\mathbf{K}$, we obtain
\begin{eqnarray}\label{eq:KT_zero}
    \mathbf{K}\mathbf{T}=(
        \nabla F_1\nabla F_1^{\mathrm{T}}
        +\nabla F_2\nabla F_2^{\mathrm{T}})\mathbf{T}=\nabla F_1(\nabla F_1\cdot\mathbf{T})+\nabla F_2(\nabla F_2\cdot\mathbf{T})=0\, .
\end{eqnarray}
Thus, $\mathbf{T}$ belongs to the null space of $\mathbf{K}$.

Let $\widehat{\mathbf{u}}$ be an arbitrary unit vector. It admits the orthogonal decomposition
\begin{eqnarray}\label{eq:u_decomposition}
    \widehat{\mathbf{u}}=a\mathbf{T}+\mathbf{u}_{\perp}\, ,
    \qquad \mathbf{u}_{\perp}\cdot\mathbf{T}=0\, ,
\end{eqnarray}
where
\begin{eqnarray}\label{eq:u_normalization}
    a^2+\|\mathbf{u}_{\perp}\|^2=1\, .
\end{eqnarray}
Substituting Eq. \eqref{eq:u_decomposition} into the quadratic form gives
\begin{eqnarray}
    \widehat{\mathbf{u}}^{\mathrm{T}}
    \mathbf{K}\widehat{\mathbf{u}}=(
        a\mathbf{T}+\mathbf{u}_{\perp})^{\mathrm{T}}
    \mathbf{K}(a\mathbf{T}+\mathbf{u}_{\perp})=
    a^2\mathbf{T}^{\mathrm{T}}\mathbf{K}\mathbf{T}+a\mathbf{T}^{\mathrm{T}}\mathbf{K}\mathbf{u}_{\perp}+a\mathbf{u}_{\perp}^{\mathrm{T}}\mathbf{K}\mathbf{T}
    +\mathbf{u}_{\perp}^{\mathrm{T}}
    \mathbf{K}
    \mathbf{u}_{\perp}\, .
\end{eqnarray}
Since $\mathbf{K}\mathbf{T}=0$ and $\mathbf{K}$ is symmetric, all terms containing $\mathbf{T}$ vanish. Hence,
\begin{eqnarray}\label{eq:Q_normal_component}
    \widehat{\mathbf{u}}^{\mathrm{T}}
    \mathbf{K}\widehat{\mathbf{u}}=\mathbf{u}_{\perp}^{\mathrm{T}}\mathbf{K}\mathbf{u}_{\perp}\, .
\end{eqnarray}

For $\mathbf{u}_{\perp}\neq0$, define
\begin{eqnarray}
    \widehat{\mathbf{v}}=\frac{\mathbf{u}_{\perp}}{\|\mathbf{u}_{\perp}\|}\, .
\end{eqnarray}
Then $\widehat{\mathbf{v}}$ is a unit vector in the normal plane and
\begin{eqnarray}
    \mathbf{u}_{\perp}=\sqrt{1-a^2}\,
    \widehat{\mathbf{v}}\, .
\end{eqnarray}
Therefore,
\begin{eqnarray}\label{eq:Q_factor}
    \widehat{\mathbf{u}}^{\mathrm{T}}\mathbf{K}\widehat{\mathbf{u}}=(1-a^2)\widehat{\mathbf{v}}^{\mathrm{T}}\mathbf{K}\widehat{\mathbf{v}}\, .
\end{eqnarray}

Because $\mathbf{K}$ is positive semidefinite,
\begin{eqnarray}
    \widehat{\mathbf{v}}^{\mathrm{T}}
    \mathbf{K}
    \widehat{\mathbf{v}}
    \geq 0\, .
\end{eqnarray}
At a regular point, $\nabla F_1$ and $\nabla F_2$ are linearly independent, so $\mathbf{K}$ has two positive eigenvalues in the normal plane. Consequently, its largest eigenvalue satisfies $\lambda_{+}>0$. If $a\neq0$, then $1-a^2<1$, and Eq. \eqref{eq:Q_factor} implies
\begin{eqnarray}
    \widehat{\mathbf{u}}^{\mathrm{T}}\mathbf{K}\widehat{\mathbf{u}}
    \leq(1-a^2)\lambda_{+}<\lambda_{+}\, .
\end{eqnarray}
Hence, no unit vector with a nonzero tangential component can maximize the
quadratic form. A maximizing direction must therefore satisfy
\begin{eqnarray}
    a=0\, ,
\end{eqnarray}
and thus
\begin{eqnarray}
    \widehat{\mathbf{u}}_{\max}
    =
    \mathbf{u}_{\perp}\, ,
    \qquad
    \widehat{\mathbf{u}}_{\max}\cdot\mathbf{T}=0\, .
\end{eqnarray}
Therefore, every maximum-sensitivity direction lies in the normal plane of the exceptional line.

Finally, since $\mathbf{K}$ is a real symmetric matrix, the Rayleigh-Ritz theorem gives
\begin{eqnarray}
    \max_{\|\widehat{\mathbf{u}}\|=1}\widehat{\mathbf{u}}^{\mathrm{T}}\mathbf{K}\widehat{\mathbf{u}}=\lambda_{+}\, .
\end{eqnarray}
The maximum is attained precisely by normalized eigenvectors satisfying
\begin{eqnarray}
    \mathbf{K}\widehat{\mathbf{u}}_{\max}=\lambda_{+}\widehat{\mathbf{u}}_{\max}\, .
\end{eqnarray}
Since the fourth-root function is strictly increasing on $[0,\infty)$, it follows that
\begin{eqnarray}
    C_{\max}=(\lambda_{+})^{1/4}\, .
\end{eqnarray}
This completes the proof.
\end{proof}

\section{Two eigenvalues with their numerical evaluations}\label{app: two_eigenvalues}
In this appendix, we give the expressions $\lambda_{+}$ and $\lambda_{-}$ of two eigenvalues of the matrix $\mathbf{K}$ and their numerical evaluations. Define $\mathbf{a}:=\nabla F_1(\mathbf{p}_0)$ and $\mathbf{b}:=\nabla F_2(\mathbf{p}_0)$. Then one has $\mathbf{K}=\mathbf{a}\mathbf{a}^{\mathrm{T}}+\mathbf{b}\mathbf{b}^{\mathrm{T}}$. The two nonzero eigenvalues of $\mathbf{K}$ are
\begin{eqnarray}\label{eq:lambda_pm_explicit}
    \lambda_{\pm}=\frac{1}{2}\Big[
        \|\mathbf{a}\|^2+\|\mathbf{b}\|^2\pm\sqrt{(\|\mathbf{a}\|^2-\|\mathbf{b}\|^2)^2+4(\mathbf{a}\cdot\mathbf{b})^2}
    \Big]\, ,
\end{eqnarray}
where $\|\cdot\|$ is the Euclidean norm.

Since $\mathbf{a}=(\operatorname{Re}h_{\mu},\operatorname{Re}h_{\nu},\operatorname{Re}h_{\xi})$, and $\mathbf{b}=(\operatorname{Im}h_{\mu},\operatorname{Im}h_{\nu},\operatorname{Im}h_{\xi})$, we should calculate gradient of function $h$ on EL. Based on Eq. (\ref{squared_difference_omega}), around the EL, we locally write the squared eigenvalue difference in $3$-dimensional parameter space $(\mu,\nu,\xi)$ as
\begin{eqnarray}\label{h_function_3D}
    h(\mu,\nu,\xi)&=&D(\xi)+E(\xi)\Big[\mu-\mu_0(\xi)\Big]+F(\xi)\Big[\nu-\nu_0(\xi)\Big]\nonumber\\
    &&+G(\xi)\Big[\mu-\mu_0(\xi)\Big]^2+H(\xi)\Big[\mu-\mu_0(\xi)\Big]\Big[\nu-\nu_0(\xi)\Big]+I(\xi)\Big[\nu-\nu_0(\xi)\Big]^2\, .
\end{eqnarray}
From the above equation, we have the derivatives of $h(\mu,\nu,\xi)$, namely
\begin{eqnarray}\label{h_mu}
    h_{\mu}=E(\xi)+2G(\xi)\Big[\mu-\mu_0(\xi)\Big]+H(\xi)\Big[\nu-\nu_0(\xi)\Big]\, ,
\end{eqnarray}
\begin{eqnarray}\label{h_nu}
    h_{\nu}=F(\xi)+H(\xi)\Big[\mu-\mu_0(\xi)\Big]+2I(\xi)\Big[\nu-\nu_0(\xi)\Big]\, ,
\end{eqnarray}
and
\begin{eqnarray}\label{h_xi}
    h_{\xi}&=&D^{\prime}+E^{\prime}\Big[\mu-\mu_0(\xi)\Big]+F^{\prime}\Big[\nu-\nu_0(\xi)\Big]+G^{\prime}\Big[\mu-\mu_0(\xi)\Big]^2+H^{\prime}\Big[\mu-\mu_0(\xi)\Big]\Big[\nu-\nu_0(\xi)\Big]+I^{\prime}\Big[\nu-\nu_0(\xi)\Big]^2\nonumber\\
    &&-\mu_0^{\prime}(\xi)\Bigg\{E+2G\Big[\mu-\mu_0(\xi)\Big]+H\Big[\nu-\nu_0(\xi)\Big]\Bigg\}-\nu_0^{\prime}(\xi)\Bigg\{F+H\Big[\mu-\mu_0(\xi)\Big]+2I\Big[\nu-\nu_0(\xi)\Big]\Bigg\}\, .
\end{eqnarray}

From Eq. (\ref{h_mu}) and Eq. (\ref{h_nu}), on each EPs $(\mu_{\text{EP},i},\nu_{\text{EP},i},\xi_i)$ of EL, one has
\begin{eqnarray}
    h_{\mu,i}=E(\xi_i)+2G(\xi_i)\Big[\mu_{\text{EP},i}-\mu_0(\xi_i)\Big]+H(\xi_i)\Big[\nu_{\text{EP},i}-\nu_0(\xi_i)\Big]\, ,
\end{eqnarray}
and
\begin{eqnarray}
    h_{\nu,i}=F(\xi_i)+H(\xi_i)\Big[\mu_{\text{EP},i}-\mu_0(\xi_i)\Big]+2I(\xi_i)\Big[\nu_{\text{EP},i}-\nu_0(\xi_i)\Big]\, .
\end{eqnarray}
Although Eq. (\ref{h_xi}) is correct, numerically it is necessary to calculate the derivatives of eight functions separately, and the errors of these derivatives will collectively enter $h_{\xi}$. A more stable method is to directly utilize the following layer local functions,
\begin{eqnarray}\label{h_i_mu_nu}
    h_i(\mu,\nu)&=&D(\xi_i)+E(\xi_i)\Big[\mu-\mu_0(\xi_i)\Big]+F(\xi_i)\Big[\nu-\nu_0(\xi_i)\Big]\nonumber\\
    &&+G(\xi_i)\Big[\mu-\mu_0(\xi_i)\Big]^2+H(\xi_i)\Big[\mu-\mu_0(\xi_i)\Big]\Big[\nu-\nu_0(\xi_i)\Big]+I(\xi_i)\Big[\nu-\nu_0(\xi_i)\Big]^2\, .
\end{eqnarray}
In order to get $h_\xi$, the fourth-order finite difference approximation for $h_\xi(\mu_{\text{EP},i},\nu_{\text{EP},i},\xi_i)$ is used from the above layer local functions. In summary, for the sake of getting the gradient of function $h$, we need to know the values of $\mu_{\text{EP},i}$, $\nu_{\text{EP},i}$, $D(\xi_i)$, $E(\xi_i)$, $F(\xi_i)$, $G(\xi_i)$, $H(\xi_i)$, $I(\xi_i)$, $\mu_0(\xi_i)$ and $\nu_0(\xi_i)$ at each $\xi_i$. Note that one must keep $\mu=\mu_{\text{EP}}$ and $\nu=\nu_{\text{EP}}$ fixed for evaluations of $h_\xi$ at each $\xi_i$. 

\bibliography{reference}
\bibliographystyle{apsrev4-1}

\end{document}